\documentclass[final,12pt]{clear2025}  

\def\given{\,|\,}
\def\biggiven{\,\big{|}\,}

\def\E{{ E}}
\def\P{{ P}}
\def\Q{{\mathrm Q}}
\def\Var{{\mathrm{Var}}}
\def\d{{ d}}

\newcommand{\abs}[1]{\left| #1 \right|}

\newcommand{\norm}[1]{\left\lVert #1 \right\rVert}

\newcommand{\ind}{1} 

\newcommand\yestag{\addtocounter{equation}{1}\tag{\theequation}}

\DeclareMathOperator*{\argmin}{argmin}

\DeclareMathOperator*{\esssup}{ess\,sup}
\DeclareMathOperator*{\essinf}{ess\,inf}

\usepackage{enumitem} 
\usepackage{mathtools}
\usepackage[capitalize,noabbrev]{cleveref}
\usepackage{dsfont}
\usepackage{stmaryrd}
\usepackage{multirow}
\RequirePackage{algorithm}
\RequirePackage{algorithmic}

\title[Stabilized Inverse Probability Weighting via Isotonic Calibration]{Stabilized Inverse Probability Weighting via Isotonic Calibration}
\usepackage{comment}
\usepackage{times}

\clearauthor{%
 \Name{Lars {van der Laan}} \Email{lvdlaan@uw.edu}\\
 \addr Department of Statistics, University of Washington, Seattle, WA 98195, USA
 \AND
 \Name{Ziming Lin} \Email{zmlin@uw.edu}\\
 \addr Department of Statistics, University of Washington, Seattle, WA 98195, USA
 \AND
 \Name{Marco Carone} \Email{mcarone@uw.edu}\\
 \addr Department of Biostatistics, University of Washington, Seattle, WA 98195, USA%
 \AND
 \Name{Alex Luedtke} \Email{aluedtke@uw.edu}\\
 \addr Department of Statistics, University of Washington, Seattle, WA 98195, USA%
}

\begin{document}

\maketitle

\begin{abstract}%
  Inverse weighting with an estimated propensity score is widely used by estimation methods in causal inference to adjust for confounding bias. However, directly inverting propensity score estimates can lead to instability, bias, and excessive variability due to large inverse weights, especially when treatment overlap is limited. In this work, we propose a post-hoc calibration algorithm for inverse propensity weights that generates well-calibrated, stabilized weights from user-supplied, cross-fitted propensity score estimates. Our approach employs a variant of isotonic regression with a loss function specifically tailored to the inverse propensity weights. Through theoretical analysis and empirical studies, we demonstrate that isotonic calibration improves the performance of doubly robust estimators of the average treatment effect.
\end{abstract}

\begin{keywords}%
   inverse probability weighting, calibration, isotonic regression, propensity, balancing%
\end{keywords}

\section{Introduction}

\subsection{Background}

Estimation and inference tasks in causal inference often involve inverse weighting by an estimate of the propensity score --- the probability of receiving treatment given a set of covariates \citep{rosenbaum1983central}. This includes, for example, prominent methods for inference of counterfactual means and average treatment effects such as inverse probability weighting (IPW) \citep{rosenbaum1983central} and doubly-robust estimation strategies like augmented inverse probability weighting \citep{robinsCausal, bang2005doubly}, targeted maximum likelihood estimation \citep{vanderLaanRose2011}, and double machine learning \citep{DoubleML}.  Despite their wide usage and theoretical motivation, these approaches are prone to instability, bias, and excessive variability induced by large inverse propensity weights, especially in settings with limited treatment overlap \citep{wang2006diagnosing, crump2009dealing, petersen2012diagnosing, zhou2020propensity}. As a result, there is a growing interest in developing robust methods to estimate inverse propensity weights, aiming to mitigate the influence of extreme values in these weights. To address this issue, various estimation strategies have been proposed, which can be broadly categorized into two groups: direct and inversion-based methods.

Inversion methods aim to construct an improved estimator of inverse propensity weights by applying post-hoc adjustments to an initial propensity score estimator before inverting it. One advantage of inversion methods is their convenience, as the initial propensity score estimator can be obtained using any off-the-shelf regression algorithm. Notable inversion methods include propensity score truncation (or clipping) and propensity score trimming. Propensity score truncation constrains the estimated propensity scores within a range bounded away from 0 and 1 \citep{wang2006diagnosing, bembom2008data, cole2008constructing, petersen2012diagnosing, gruber2022data}. Propensity score trimming is a related approach that involves excluding observations with extreme propensity score weights, i.e., assigning them zero weight \citep{imbens2004nonparametric, crump2009dealing, petersen2012diagnosing}. The choice of truncation or trimming threshold can be predetermined or data-adaptive. Pre-specified thresholds, however, may yield inconsistent estimates of propensity score weights, leading to biased estimations of treatment effects. To address this bias and improve the reliability of inferences, data-adaptive methods for threshold selection, such as collaborative targeted maximum likelihood estimation, have been proposed \citep{bembom2008data, ju2019adaptive, gruber2022data}. While inversion methods can lead to improved inference, they are limited by the fact that even when the propensity score estimator is properly truncated or trimmed, it may not serve as a good estimator of the true inverse propensity weights.

Direct learning methods improve upon inversion methods by using supervised statistical learning tools to directly estimate the inverse propensity weights through the optimization of a specifically tailored objective function. Notable direct learning methods include entropy \citep{hainmueller2012entropy} and covariate balancing weights \citep{imai2014covariate, zubizarreta2015stable, zhao2019covariate, wang2020minimal}, minimax linear estimation \citep{hirshberg2019minimax, hirshberg2021augmented}, and automatic debiased machine learning \citep{chernozhukov2018double, chernozhukov2021automatic, chernozhukov2022riesznet, chernozhukov2022automatic,bruns2023augmented}. These approaches have been demonstrated, in simulations, to yield less biased and more stable parameter estimates with improved confidence interval coverage compared to inversion methods relying on propensity score estimation \citep{zhao2019covariate, chernozhukov2021automatic}. Even so, when estimated using aggressive machine learning algorithms, such as neural networks or gradient-boosted trees, these inverse weights may still exhibit poor calibration \citep{guo2017calibration}. Furthermore, the practical applicability of direct learning methods is hindered by their dependence on nonstandard loss functions, which are not readily available in standard software packages for machine learning. As a consequence of this limitation, direct methods may be employed in conjunction with parametric models, which may lack the necessary flexibility to accurately estimate inverse propensity weights. While some machine learning algorithms, such as \texttt{xgboost} gradient-boosted trees \citep{chen2016xgboost} and neural networks \citep{pedregosa2011scikit}, provide flexibility to accommodate customized loss functions, the extent of customization, scalability, and numerical efficiency can vary among commonly used learning algorithms and software choices. 

\subsection{Contributions of this work}

In this work, we introduce a distribution-free approach for calibrating inverse propensity weights directly from user-supplied propensity score estimates. We propose a novel algorithm, isotonic calibrated inverse probability weighting (IC-IPW), which uses isotonic regression to transform cross-fitted propensity score estimates into well-calibrated inverse propensity weights. Departing from traditional inversion methods, IC-IPW learns an optimal monotone transformation of the propensity score estimator through a variant of isotonic regression that minimizes a loss function tailored to the inverse propensity score. This approach is computationally efficient, computable in linear time, and easily implemented using standard isotonic regression software. It retains the flexibility and simplicity of inversion-based estimation while providing robustness in settings with limited overlap. To demonstrate the advantages of IC-IPW, we examine its use with augmented inverse probability weighted (AIPW) estimators for estimating the average treatment effect (ATE). We show that IC-IPW can relax the conditions required for achieving asymptotic linearity and nonparametric efficiency of AIPW, while also improving empirical performance in terms of bias and coverage.

Our work contributes to the recent studies of \cite{gutman2022propensity}, \cite{deshpande2023calibrated}, and \citet{ballinari2024improvingfinitesampleperformance} on the benefits of calibrated propensity score estimators for causal inference applications. In their empirical study, \cite{gutman2022propensity} found Platt’s scaling to be beneficial for calibrating propensity score estimators in causal inference. Additionally, \cite{deshpande2023calibrated} propose using propensity score calibration in IPW and AIPW estimators and provide theoretical guarantees for kernel density-based calibration using a log-likelihood loss. Similarly, \citet{ballinari2024improvingfinitesampleperformance} proposed using propensity score calibration within the double machine learning framework to improve the finite-sample performance of ATE estimators. They examined the empirical performance of various probability calibration procedures, including Platt’s and temperature scaling and a form of isotonic calibration that differs from ours. Distribution-free calibration guarantees for isotonic calibration of regression and conditional average treatment effect functions were established in \cite{van2023causal}. We build on this work by establishing that distribution-free calibration guarantees for isotonic calibration can be extended to the inverse propensity score based on a tailored loss function.

Our work differs from prior studies of propensity score calibration in several key aspects. First, differing from prior work, we propose calibrating the inverse propensity score rather than the propensity score itself and introduce a measure for balance-based calibration of inverse propensity weights. This distinction is crucial because calibration methods developed for the propensity score may not yield well-calibrated inverse propensity weights. In this sense, our work is closely related to \cite{deshpande2023calibrated}, which considers calibrating propensity score estimates using the chi-squared divergence loss. While the authors do not explicitly consider calibration of the inverse propensity score, the chi-squared divergence loss can be reformulated as a loss for the inverse propensity score. Second, we propose a novel calibration procedure based on isotonic regression to construct calibrated inverse propensity weights, which is free from tuning parameters. Furthermore, we formally establish that our isotonic regression procedure provides distribution-free calibration guarantees that impose no assumptions on the functional form of the data-generating distribution. In contrast, Platt scaling approach of \cite{gutman2022propensity} relies on parametric assumptions, and the kernel-smoothing method of \cite{deshpande2023calibrated} depends on smoothness assumptions and requires tuning of the kernel bandwidth parameter. Third, in the context of debiased machine learning, we combine cross-fitting and calibration in a novel, data-efficient manner. Specifically, we apply calibration once to cross-fitted propensity score estimators, unlike the approach of \citet{ballinari2024improvingfinitesampleperformance}, which repeats the calibration step across multiple data splits.

In our companion paper on calibrated debiased machine learning \citep{van2024automatic}, we establish the role of nuisance function calibration in constructing doubly robust asymptotically linear estimators, thereby providing not only doubly robust consistency but also facilitating doubly robust inference (e.g., confidence intervals and hypothesis tests). Notably, Theorem 4 in \cite{van2024automatic} shows that IPW estimators of linear functionals, such as the ATE, with weights obtained using our proposed IC-IPW procedure are asymptotically linear, even when the initial propensity score estimator is derived using flexible machine learning tools.

 The outline of this manuscript is as follows. Our proposed IC-IPW algorithm is outlined in Section \ref{section::alg} and its theory is presented in Section \ref{section::theory}. In Section \ref{section::theory::cal}, we give calibration and mean square error guarantees for IC-IPW, and in Section \ref{section::theory::ATE} we theoretically illustrate the benefits of IC-IPW for ATE estimation using AIPW estimator. Finally, Section \ref{section::simulation} presents simulation experiments to evaluate the estimator's performance under different levels of overlap.

\section{Methodology: stabilized inverse weighting via isotonic calibration}
\label{section::alg}

\subsection{Notation}
\label{section::setup}

Suppose we observe \( n \) independent and identically distributed observations, \( O_1, \dots, O_n \), of the data structure \( O \coloneq (W, A, Y) \) drawn from a probability distribution \( P_0 \). In this data structure, \( W \in \mathbb{R}^d \) is a vector of baseline covariates, \( A \in \mathcal{A} \subset \mathbb{R} \) is a discrete treatment assignment taking values in \( \mathcal{A} \), and \( Y \in \mathbb{R} \) is a bounded outcome. Without loss of generality, we assume \( \{0,1\} \subset \mathcal{A} \), where \( \{0\} \) represents a reference treatment level and \( \{1\} \) a treatment level of interest. For a given distribution \( \P \) and realization \( (a, w) \) of \( (A, W) \), we denote the (generalized) propensity score by \( \pi_{P}(a \mid w) \coloneq \P(A = a \mid W = w) \) and the outcome regression by \( \mu_{P}(a, w) \coloneq E_P[Y \mid A = a , W = w] \). We also denote the inverse propensity score by \( \alpha_{P}(a, w) \coloneq 1 / \pi_{P}(a \mid w) \) and the conditional average treatment effect (CATE) function by \( \tau_{P}(w) \coloneq \mu_{P}(1, w) - \mu_{P}(0, w) \). To simplify notation, we write \( S_0 \) for summaries \( S_{P_0} \) of the true distribution \( P_0 \). Throughout this text, for a function \( f \), we use \( f^{-1} \) to denote the reciprocal function, defined pointwise as \( f^{-1}(x) := \frac{1}{f(x)} \). All essential suprema are taken with respect to \( P_0 \) (or its marginal distributions).


\subsection{Balance-based calibration for inverse propensity weights}
\label{sec::balcal}

For a treatment level \( a_0 \in \mathcal{A} \), let \( \pi_n(a_0 \mid \cdot) \) be an arbitrary estimator of \( \pi_0(a_0 \mid \cdot) \), obtained using, for example, flexible statistical learning tools. Specifically, \( \pi_n(a_0 \mid \cdot) \) could be derived by nonparametrically regressing the treatment indicator \( \ind(A = a_0) \) on covariates. Alternatively, we could set \( \pi_n(a_0 \mid \cdot) \coloneq \alpha_n^{-1}(a_0 \mid \cdot) \), where \( \alpha_n(a_0 \mid \cdot) \) is an estimator of \( \pi_0^{-1}(a_0 \mid \cdot) \) obtained using loss functions for inverse propensity weights \citep{zhao2019covariate, chernozhukov2021automatic}.

We say an estimator $\pi_n^{-1}(a_0 \given \cdot)$ of the inverse propensity score function $\pi_0^{-1}(a_0\given \cdot)$ is \textit{perfectly calibrated} for covariate balance if, for all functions $h: \mathbb{R} \rightarrow \mathbb{R}$ in $L^2(P_0)$,
\begin{equation*}
    \int (h \circ \pi_n)(a_0 \given w)\left\{\frac{\pi_0(a_0 \given w)}{\pi_n(a_0 \given w)}-1 \right\} dP_0(w) = 0.
\end{equation*}
In words, a perfectly calibrated estimator $\pi_n^{-1}(a_0 \given \cdot)$ achieves covariate balance, in the sense of \cite{imai2014covariate} and \cite{ben2021balancing}, within subgroups defined by levels of estimated values of the propensity score. The estimator $\pi_n^{-1}(a_0 \given \cdot)$ is perfectly calibrated for covariate balance if and only if $P_0$-almost surely $\pi_n(a_0 \given W) =  \gamma_0^{(a_0)}(W, \pi_n)$, where, for any function $\pi$, we define the \textit{calibration function} $ w \mapsto \gamma_0^{(a_0)}(w, \pi)\coloneq E_0[\pi_0(a_0 \given W) \given \pi(a_0\given W) = \pi(a_0 \given w)]$. 

As a measure of deviation from perfect calibration, we consider the following $\chi^2$-squared divergence measure \citep{pearson1900x} of calibration error for inverse propensity weights:
\begin{equation}
    {\rm CAL}^{(a_0)}(\pi^{-1}) \coloneq \int \left\{\frac{\gamma_0^{(a_0)}(w, \pi)}{\pi(a_0 \given w)} -1 \right\}^2 d P_{0}(w). \label{eqn::calerror}
\end{equation}
Intuitively, for an inverse propensity score estimator $\pi_n^{-1}(a_0 \mid \cdot)$ that is well-calibrated in the sense that ${\rm CAL}^{(a_0)}(\pi_n^{-1})$ is small, a large inverse propensity weight should only be assigned to an individual when the true propensity score $\pi_0(a_0 \mid \cdot)$ is, on average, small among individuals with that weight. Consequently, in scenarios with limited overlap, a well-calibrated estimator should avoid assigning excessively large weights to observations. Notably, \cite{deshpande2023calibrated} established that calibration with respect to this $\chi^2$-divergence measure leads directly to smaller estimation errors for IPW and AIPW estimators of the ATE (see Lemma 3.3 and Theorem 3.5 of \cite{deshpande2023calibrated}). Therefore, developing calibration procedures that guarantee inverse propensity score estimators with small calibration error under this measure is of great interest for causal inference.



We note that \eqref{eqn::calerror} differs from the $\ell^2$-integrated calibration error for the propensity score, commonly used in probability calibration \citep{gupta2021distribution}, which is given by
\begin{equation}
    \int \left\{\gamma_0(a_0)(w, \pi) - \pi(a_0 \mid w) \right\}^2 dP_0(w).
    \label{eqn::calerrorbad}
\end{equation} 
Notably, a propensity score estimator $\pi_n$ that is well-calibrated with respect to this calibration measure may not lead to well-calibrated inverse propensity weights in the sense of \eqref{eqn::calerror}. This discrepancy arises because the calibration error of \eqref{eqn::calerrorbad} is insensitive to scenarios with limited overlap, where the inverse propensity score $\pi_n^{-1}(a_0 \mid w)$ may take large values. Distribution-free calibration guarantees for isotonic calibration of regression functions (such as the propensity score) with respect to this measure are a consequence of \cite{van2023causal}.


\subsection{Isotonic calibrated inverse probability weighting}
\label{sec::ICIPWpart}
In this section, we introduce our novel calibration algorithm, IC-IPW, which uses isotonic regression to construct inverse propensity weights calibrated according to the definition in \eqref{eqn::calerror}.

IC-IPW is motivated by the observation that inversion strategies, such as propensity score truncation, commonly apply a nonincreasing monotone transformation to a propensity score estimator \( \pi_n(a_0 \mid \cdot) \) to derive an inverse propensity score estimator. The simplest example of this is the estimator \( \pi_n^{-1}(a_0 \mid \cdot) \), corresponding to the inversion map \( x \mapsto x^{-1} \). Another option is the truncated estimator given by $\{(c_n \vee \pi_n(a_0 \mid \cdot)) \wedge (1 - c_n)\}^{-1}$, where $c_n \in (0, 1/2]$ is a possibly data-adaptive threshold, corresponding to the transformation $x \mapsto \{(c_n \vee x) \wedge (1 - c_n)\}^{-1}$. The optimal nonincreasing transformation $f_{n,0}^{(a_0)}: \mathbb{R} \rightarrow \mathbb{R}$ of $\pi_n(a_0 \mid \cdot)$, which minimizes the mean squared error for $(w, a) \mapsto \ind(a = a_0) \pi_0^{-1}(a_0 \mid \cdot)$, is given by
\begin{align*}
    f_{n,0}^{(a_0)} &\in \argmin_{\theta \in \mathcal{F}_{\rm anti}} E_0 \big[ \ind(A = a_0) \{(\theta \circ \pi_n)(a_0 \mid W) - \pi_0^{-1}(a_0 \mid W)\}^2 \big] \\
    &= \argmin_{\theta \in \mathcal{F}_{\rm anti}} E_0 \big[ \ind(A = a_0)(\theta \circ \pi_n)(a_0 \mid W)^2 - 2 (\theta \circ \pi_n)(a_0 \mid W) \big].
\end{align*}
Notably, the mapping $(w, a) \mapsto f_{n,0}^{(a)} \circ \pi_n(a \mid w)$ is a mean-squared optimal estimator of the inverse propensity score $(w, a) \mapsto \pi_0^{-1}(a \mid w)$.


Our proposed IC-IPW estimator of $\pi_0^{-1}(a_0 \mid \cdot)$ is given by $\alpha_n^*(a_0 \given \cdot) = (f_n^{(a_0)} \circ \pi_n)(a_0 \given \cdot) $ for a data-dependent monotone nonincreasing transformation $f_n^{(a_0)}: \mathbb{R} \rightarrow \mathbb{R} \cup \{\infty\}$ obtained via antitonic empirical risk minimization as:
\begin{equation}
 f_n^{(a_0)}  \in \argmin_{\theta \in \mathcal{F}_{\rm anti}} \sum_{i=1}^n  \big[\ind(A_i = a_0) (\theta \circ \pi_n)^2(a_0 \given W_i) - 2(\theta \circ \pi_n)(a_0 \given W_i) \big],   \label{eqn::ERManti}
\end{equation} 
where $\mathcal{F}_{\rm anti} \coloneq \{\theta:\mathbb{R} \rightarrow \mathbb{R}; \; \theta \text{ is monotone  nonincreasing}\}$ denotes the space of antitonic functions. We follow \citet{groeneboom1993isotonic} in taking the unique c\`{a}dl\`{a}g piecewise constant solution to the antitonic regression problem, which only has jumps at observed values in \( \{\pi_n(a_0 \given W_i): i \in [n]\} \). The loss function \( (o, \alpha) \mapsto \ind(a = a_0) \alpha(a \mid w)^2 - 2 \alpha(a_0 \mid w) \), used in \eqref{eqn::ERManti}, has been applied for learning balancing weights \citep{zhao2019covariate} as well as for estimating the inverse propensity score \citep{chernozhukov2018double}. However, to the best of our knowledge, this loss function has not been previously utilized for isotonic regression or calibration purposes.

The function \( f_n^{(a_0)} \) may take an infinite value when there exists a boundary observation \( i(1) \in [n] \) with {\small \( \pi_n(a_0 \mid W_{i(1)}) = \min_{i \in [n]}\{\pi_n(a_0 \mid W_{i})\} \)}  such that \( A_{i(1)} \neq a_0 \). This is related to the fact that isotonic regression can exhibit poor behavior at the boundary of the input space \citep{groeneboom1993isotonic}. In the context of causal inference and debiased machine learning, infinite weights among observations \( i \) with \( A_i \neq a_0 \) are usually not problematic, as typically only observations in the relevant treatment arm with \( A_i = a_0 \) need to be reweighted using these weights. That is, only estimates of {\small\( \{ \frac{1(A_i = a_0)}{\pi_0(a_0 \mid W_i)} : i \in [n] \} \)} are required. Nonetheless, we can address this issue by applying a boundary correction that adaptively truncates \( f_n^{(a_0)} \) to lie within the range \( [1, b_n^{(a_0)}] \), where the data-dependent truncation level is given by the largest finite value, {\small \( b_n^{(a_0)} = \max_{i \in [n]: A_i = a_0} f_n^{(a_0)}(\pi_n(a_0 \mid W_i)) \)}. This truncation level preserves the value of { \small \( 1(A_i = a_0) f_n^{(a_0)}(\pi_n(a_0 \mid W_i)) \) } for all \( i \in [n] \), while ensuring that the weight estimates {\small \( f_n^{(a_0)}(\pi_n(a_0 \mid W_i)) \)} with \( A_i \neq a_0 \) remain finite. Alternatively, one can constrain the minimum number of observations in each constant segment of the antitonic solution, ensuring that at least one treated observation lies in each level set of \( f_n^{(a_0)} \). We provide \texttt{R} and \texttt{Python} code implementing both the truncation and this constraint in Appendix \ref{sec:code}, leveraging a constrained implementation of isotonic regression in \texttt{xgboost} \citep{chen2016xgboost}.

The first-order conditions characterizing the minimizing solution imply that the proposed estimator $\alpha_n^*(a_0 \mid \cdot)$ solves, for any map $h: \mathbb{R} \cup \{\infty\} \rightarrow \mathbb{R}$ with $h(\infty) := 0$, the covariate balancing equation:
\begin{equation}
    \frac{1}{n} \sum_{i=1}^n (h \circ \alpha_n^*)(a_0 \mid W_i) \left\{\ind(A_i = a_0)\alpha_n^*(a_0 \mid W_i) - 1\right\} = 0. \label{firstorderCond}
\end{equation}
Hence, the calibrated weight estimator $\alpha_n^*(a_0 \mid \cdot)$ achieves exact empirical covariate balance within levels of the estimated weights. Notably, by choosing \( h \) as a level set indicator, we can show that \(\frac{1}{N_t} \sum_{i \in I_t} \ind(A_i = a_0) \alpha_n^*(a_0 \mid W_i) = 1\)
for all \( t \in \mathbb{R} \), where \( N_t = \sum_{i=1}^n 1(\alpha_n^*(a_0 \mid W_i) = t) \) and \( I_t = \{ i : \alpha_n^*(a_0 \mid W_i) = t \} \). This indicates that the calibrated weight estimators are automatically stabilized among level sets of the estimated weights \citep{robins2000marginal}.
As a consequence, the IC-IPW estimator is piecewise constant, where each unique value corresponds to the inverse of an empirical mean over a subset of $\{1(A_1=a_0), \dots, 1(A_n=a_0)\}$. Specifically, for each observation index $i_0 \in [n]$, there exist endpoints $\ell_0 \leq m_0$ that depend on the sample size $n$, such that
\[
\alpha_{n}^*(a_0 \mid W_{i_0}) = \frac{m_0 - \ell_0 + 1}{\sum_{\ell_0 \leq r \leq m_0} 1(A_{i(r)} = a_0)},
\]
where $(i(r): r \in [n])$ is an appropriate permutation of $[n]$.

We can show that the antitonic transformation satisfies \( f_n^{(a_0)} =  1/g_n^{(a_0)} \), where \( g_n^{(a_0)} \) is the isotonic calibrator of the propensity score \( \pi_0(a_0 \mid \cdot) \) obtained via isotonic regression as:
\begin{equation}
    g_n^{(a_0)}  \in \argmin_{\phi \in \mathcal{F}_{\rm iso}} \sum_{i=1}^n  \big\{\ind(A_i = a_0) - (\phi \circ \pi_n)(a_0 \mid X_i) \big\}^2,  \label{eqn::ERMiso}
\end{equation} 
where \( \mathcal{F}_{\rm iso} \coloneq \{\theta:\mathbb{R} \rightarrow \mathbb{R}; \; \theta \text{ is monotone nondecreasing}\} \) denotes the space of isotonic functions.
Thus, the optimization problem given by \eqref{eqn::ERManti} can be computed via isotonic probability calibration for \( \pi_0(a_0 \given \cdot) \) using \( \ind(A = a_0) \) as the outcome \citep{niculescu2005predicting}. For a binary treatment \( A \in \{0,1\} \), it can be shown that \( g_n^{(1)} = 1 - g_n^{(0)} \), such that a single isotonic regression suffices to calibrate the propensity weights.

\begin{algorithm}[!htb]
\begin{algorithmic}[1]
\caption{Cross-fitted Isotonic-Calibrated Inverse Probability Weighting}
\label{alg:1}\vspace{.05in}
\REQUIRE   dataset $\mathcal{D}_n = \{O_i: i \in [n]\}$, \# of cross-fitting splits $k$, $a_0 \in \mathcal{A}$
\vspace{.03in}
\STATE partition $\mathcal{D}_n$ into datasets $\mathcal{C}^{(1)},\mathcal{C}^{(2)},\ldots,\mathcal{C}^{(J)}$;
\FOR {$s = 1,2,\ldots,J$}
\STATE let $j(i)=s$ for each $i\in \mathcal{C}^{(s)}$;
\STATE get initial estimator $\pi_{n,s}$ of $\pi_0$ from $\mathcal{E}^{(s)} \coloneq \mathcal{D}_n \backslash \mathcal{C}^{(s)}$.
\ENDFOR
\STATE set $f_n^{(a_0)} \coloneq 1/\{c_n^{(a_0)} \vee g_n^{(a_0)}\}$ where $c_n^{(a_0)} = \min_{i \in [n]: A_i = a_0} g_n^{(a_0)}(\pi_{n,j(i)}(a_0 \given W_i))$ and $g_n^{(a_0)}$ is obtained using isotonic regression as
\vspace{-0.1cm}
\begingroup\small
$$
    g_n^{(a_0)} \in  \argmin_{\phi \in \mathcal{F}_{\rm iso}} \sum_{i=1}^n  \big\{\ind(A_i=a_0) - (\phi \circ \pi_{n,j(i)})(a_0 \given W_i) \big\}^2.
$$
\endgroup
\vspace{-0.1cm}
\STATE set $\alpha_{n,j}^*(a_0\mid \cdot) :=  f_n^{(a_0)}(\pi_{n,j}(a_0 \given\cdot ))$ for each $j \in [J]$;
\RETURN {\small weight functions $\{\alpha_{n,j}^*(a_0\mid \cdot) : j \in [J]\}$ and estimates $\{\alpha_{n,j(i)}^*(a_0\mid W_i) : i \in [n]\}$}
\end{algorithmic}

\vspace{.05in}
\end{algorithm}

In causal inference using debiased and targeted machine learning, cross-fitting inverse propensity weights is a common strategy to relax conditions on estimator complexity and mitigate overfitting \citep{van2011cross, DoubleML}. Similarly, to ensure calibration guarantees under weak conditions, we recommend applying cross-fitting when fitting the initial propensity score estimator. Algorithm \ref{alg:1} outlines the procedure for IC-IPW with cross-fitting and an adaptive truncation-based boundary correction. Specifically, it involves partitioning the available data into $J$ folds, where each fold computes an initial propensity score estimator $\pi_{n,j}(a_0 \given \cdot)$ using only the data from the complementary folds. The resulting out-of-fold estimates from the $k$ estimators are then combined to obtain $n$ propensity estimates for each observation. Akin to \eqref{eqn::ERMiso}, these out-of-fold estimates can be used to learn a single nonincreasing transformation $f_n^{(a_0)}: \mathbb{R} \rightarrow \mathbb{R}$ such that $\{f_n^{(a_0)} \circ \pi_{n,j}(a_0 \given \cdot): j \in [J] \}$ are isotonic calibrated cross-fitted estimators.  

We note that the isotonic fitting step of Algorithm \ref{alg:1} is deliberately not cross-fitted so that all available data are used to calibrate the inverse weight estimates. Using all available data for calibration typically improves performance and stability and, as we will show, does not harm the validity of debiased machine learning estimators since the class of isotonic functions is Donsker. In the context of AIPW estimation, \citet{ballinari2024improvingfinitesampleperformance} do not follow this approach and instead calibrate propensity scores using cross-fitted isotonic calibration. When predicted out-of-sample, estimated weights in IPW and AIPW estimators obtained by naively inverting isotonic calibrated propensity scores may have infinite values due to calibrated probabilities being equal to \(0\) or \(1\). Consequently, as found experimentally by \citet{ballinari2024improvingfinitesampleperformance}, isotonic regression may produce infinite weights when calibration is cross-fitted—a problem not encountered by our approach. We note that if cross-fitting is desired, this issue can be addressed using our data-adaptive truncation procedure, as outlined in Algorithm \ref{alg:1}.

\section{Theoretical properties}
\label{section::theory}
\subsection{Calibration and estimation error guarantees}
 \label{section::theory::cal}

We will now demonstrate that the $\chi^2$-calibration error, as defined in Section \ref{sec::balcal}, of our cross-fitted IC-IPW estimators proposed in Algorithm \ref{alg:1} tends to zero at a fast rate. Furthermore, we will show that IC-IPW generally improves the mean square error for the inverse propensity weights compared to the uncalibrated cross-fitted weights. Consequently, in the pursuit of obtaining calibrated weight estimates, IC-IPW typically does not harm the performance of the user-supplied weight estimators.

We will make use of the following conditions. For each $j \in [J]$, we denote the fold-specific IC-IPW estimator $\alpha_{n,j}^*$ by the mapping $(w,a) \mapsto f_n^{(a)} \circ \pi_{n,j}(a \mid w)$, where $f_n^{(a)}$ is the calibrator obtained from Algorithm \ref{alg:1}.

\begin{enumerate}[label=\bf{(C\arabic*)}, ref={C\arabic*}, series=cond]
 \item \textit{Estimator boundedness:} There exists $M < \infty$ such that, for all $j \in [J]$, $\esssup_{a,w} |\pi_{n,j}(a \given w)| < M$ with probability tending to one. 
    \label{cond::bound}
    \item \textit{Positivity:}  \label{cond::positivity2} \label{cond::positivity} There exists $0 < \eta < 1$ such that, for all $j \in [J]$, $\esssup_{a,w} |\alpha_{n,j}^*(a \given w)| < \eta^{-1}$ and $\esssup_{a,w} \lvert \pi_0^{-1}(a \given w) \rvert < \eta^{-1}$ with probability tending to one. 
     \item \textit{Limited boundary observations:} \label{cond::boundary}  $\#\{i \in [n]: A_i \neq a_0 \text{ and } \pi_{n,j(i)}(a_0 \mid W_{i}) < \min_{k \in [n]: A_k = a_0} \pi_{n,j(i)}(a_0 \mid W_{k})\}= O_p(n^{1/3}).$
   \item \textit{Calibration function has finite total variation:} There exists an $M < \infty$ such that, for all $j \in [J]$, $n \in \mathbb{N}$, and $a_0 \in \mathcal{A}$, the function $g_{0,n,j}^{(a_0)}: \mathbb{R} \to \mathbb{R}$, defined by $g_{0,n,j}^{(a_0)} \circ \pi_{n,j} = \gamma_0^{(a_0)}(\cdot, \pi_{n,j})$, has a total variation norm that is almost surely bounded by $M$.
\label{cond::variation}
\end{enumerate}
The following theorem establishes bounds for the calibration error of the IC-IPW estimator, and its proof follows from the empirical balancing property given by Equation \eqref{firstorderCond}.

 \begin{theorem}[Calibration error]
\label{theorem::CAL} Under Conditions \ref{cond::bound}-\ref{cond::variation}, for each $a_0 \in \mathcal{A}$, it holds
\begin{align*}
    \frac{1}{J}\sum_{i=1}^J {\rm CAL}^{(a_0)}(\alpha_{n,j}^*) = O_{p}(n^{-2/3}) .
\end{align*}  
\end{theorem}

Theorem \ref{theorem::CAL} provides distribution-free calibration guarantees for the IC-IPW estimator under weak conditions. The calibration error tends to zero at the rate of $n^{-2/3}$ irrespective of the smoothness of the inverse propensity score or dimension of the covariate vector. Is it interesting to note that IP weight estimators that are consistent in mean squared error are typically asymptotically calibrated. However, particularly in settings with high-dimensional covariates or nonsmooth inverse propensity scores, the calibration error rate of such estimators can be arbitrarily slow.

Condition \ref{cond::bound} is automatically satisfied when the initial propensity score estimator \( \pi_{n,j} \) takes values in the range \( [0,1] \). Condition \ref{cond::positivity2} is standard in causal inference and requires that all individuals have a positive probability of being assigned to either treatment or control. Condition \ref{cond::positivity} also ensures that the calibrated inverse propensity weights are uniformly bounded with probability approaching one, which holds as long as the truncation level \( c_n \) in Alg. \ref{alg:1} remains uniformly bounded away from zero, a condition we expect to hold under \ref{cond::positivity2}. Condition \ref{cond::boundary} is a mild regularity condition that we expect to be satisfied under \ref{cond::positivity2}, since, for a fixed function \( \pi(a_0 \mid \cdot) \), by properties of the minimum order statistic, \( \min_{k \in [n]: A_k = a_0} \pi(a_0 \mid W_{k}) \) should converge in probability to \( \essinf_w \pi(a_0 \mid w) \) at a rate of \( n^{-1} \) when \( \pi(a_0 \mid W) \) is a continuous random variable, or equal \( \essinf_w \pi(a_0 \mid w) \) with probability tending to one when it is discrete.
Condition \ref{cond::variation} excludes cases in which the best possible predictor of the propensity score $\pi_0$, given only the initial propensity score estimator $\pi_{n,j}$, has pathological behavior, in that it has infinite variation norm as a (univariate) mapping of $\pi_{n,j}$. We stress here that isotonic regression is used only as a tool for calibration, and our theoretical guarantees do not require any monotonicity on components of the data-generating mechanism --- for example, $\gamma_0(w, \pi_{n,j})$ need not be monotone as a function of $\pi_{n,j}(w)$. 
 

The following theorem establishes that IC-IPW not only constructs calibrated weight estimates but can also reduce the mean squared error of the inverse propensity weights compared to uncalibrated cross-fitted weights. In the following, we define the optimal nondecreasing transformations of the initial cross-fitted propensity score estimators as:
\begin{align*}
    \big\{f_{n,0}^{(a_0)}: a_0 \in \mathcal{A}\big\} \in \argmin_{\{f^{(a)}: a \in \mathcal{A}\} \subset \mathcal{F}_{\rm anti}} \sum_{j \in [J]} \int  \big\{  f^{(a)} \circ \pi_{n,j}(a \given w)  -   \pi_0^{-1}(a \given w) \big\}^2  dP_0(a,w).
\end{align*}

\begin{theorem}[Mean square error]
\label{theorem::MSE} Assume Conditions \ref{cond::bound}-\ref{cond::boundary} holds. As $n \rightarrow \infty$,  
\begin{align*}
    \sum_{j \in [J]} \int \big\{\alpha_{n,j}^*(a \given w) - f_{n,0}^{(a)} \circ \pi_{n,j}(a \given w)\big\}^2dP_0(a,w) = O_{p}(n^{-2/3}).
\end{align*}
As a consequence, for any nondecreasing transformations $\{f^{(a)} \in \mathcal{F}_{\rm anti}: a \in \mathcal{A} \}$,
\begin{align*}
     \frac{1}{J}\sum_{j \in [J]}  \|\alpha_{n,j}^* - \pi_{0}^{-1} \|^2 \leq  O_{p}(n^{-2/3}) +  \frac{1}{J}\sum_{j \in [J]}  \int \big\{f^{(a)} \circ \pi_{n,j}(a \given w) -  \pi_{0}^{-1}(a \given w)\big\}^2dP_0(a,w).
\end{align*} 
\end{theorem}

An important implication of Theorem \ref{theorem::MSE} is that the data-dependent transformations \( \{f_n^{(a_0)} \in \mathcal{F}_{\rm anti} : a_0 \in \mathcal{A}\} \) are nearly mean square error optimal among all nondecreasing transformations of the propensity score estimators \( \{\pi_{n,j}(a_0 \mid \cdot) : a_0 \in \mathcal{A}\} \), up to an error term of \( O_{p}(n^{-2/3}) \). This result suggests that the IC-IPW estimator performs at least as well as, and potentially better than, \( \pi_{n,j}^{-1} \) and its truncated variants. 

\subsection{Isotonic calibration and AIPW estimation of the ATE}
 \label{section::theory::ATE}

In this section, we demonstrate the advantages of IC-IPW by examining isotonic-calibrated AIPW estimators for the ATE \citep{robinsCausal, bang2005doubly}, showing how IC-IPW potentially relaxes the assumptions needed for asymptotic linearity and efficiency.
 
Let $\mu_{n,j}$ and $\pi_{n,j}$, for \( j \in [J] \), be initial cross-fitted estimators for the outcome regression $\mu_0$ and the propensity score $\pi_0$. Let $\alpha_{n,j}^*$, also indexed by \( j \in [J] \), denote the cross-fitted IC-IPW estimators of the inverse propensity score $\pi_0^{-1}$, obtained from $\pi_{n,1}, \dots, \pi_{n,J}$ via Algorithm \ref{alg:1}. Our cross-fitted IC-AIPW estimator for the ATE, $\psi_0 = E_0[\tau_0(W)]$, is defined as
$$
\psi_n^* \coloneq \frac{1}{n} \sum_{i=1}^n \left[ \mu_{n,j(i)}(1, W_i) - \mu_{n,j(i)}(0, W_i) + \delta_{\pm}(A_i) \alpha_{n,j(i)}^*(A_i \mid W_i) \{Y_i - \mu_{n,j(i)}(A_i, W_i)\} \right],
$$
where \( j(i) = s \) if observation \( O_i \) belongs to the training fold \( \mathcal{C}^{(s)} \), and \( \delta_{\pm}(a) = \ind(a=1) - \ind(a=0) \) for \( a \in \{0,1\} \). Let \( \varphi_0 : o = (w, a, y) \mapsto \mu_0(1, w) - \mu_0(0, w) - \psi_0 + \big\{ \frac{\ind(a=1)}{\pi_0(1 \mid w)} - \frac{\ind(a=0)}{\pi_0(0 \mid w)} \big\} (y - \mu_0(a, w)) \) denote the \( P_0 \)-efficient influence function of the ATE parameter. The next theorem establishes that IC-AIPW provides valid inference for the ATE.

\begin{enumerate}[label=\bf{(D\arabic*)}, ref={D\arabic*}, series=examp]
     \item \textit{Bounded estimator:} $\mu_{n,j}(A,W)$ is almost surely bounded with probability tending to one for each $j \in [J]$. \label{cond::boundedAIPW}
     \item \label{cond::AIPW} For each $a_0 \in \{0,1\}$ and $j \in [J]$, the following hold:
     \begin{enumerate}
        \item \textit{Consistency of oracle transformation:} $\lVert f_{n,0}^{(a_0)} \circ \pi_{n,j}(a_0 \given \cdot) - \pi_0^{-1}(a_0 \given \cdot)\rVert_{P_0}  = o_{p}(1)$. \label{cond::AIPW1} 
        \item \textit{Outcome regression rate:} $\lVert \mu_{n,j} - \mu_0 \rVert_{P_0} = o_{p}(n^{-1/6})$. \label{cond::AIPW2} 
        \item \textit{Doubly robust rate:} 
    $\lVert f_{n,0}^{(a_0)} \circ \pi_{n,j}(a_0 \given \cdot) - \pi_0^{-1}(a_0 \given \cdot) \rVert_{P_0} \times \lVert \mu_{n,j} - \mu_0 \rVert_{P_0} = o_{p}(n^{-1/2}).$ \label{cond::AIPWDR} 
     \end{enumerate}
\end{enumerate}
 
\begin{theorem}[Efficiency for ATE]
\label{theorem::AIPW}
Under Conditions \ref{cond::bound}-\ref{cond::positivity2} and \ref{cond::boundedAIPW}-\ref{cond::AIPW}, the IC-AIPW estimator $\psi_n^*$ satisfies $\psi_n^* = \psi_0 + (P_n - P_0) \varphi_0 + o_{p}(n^{-1/2})$ and is, thus, regular, asymptotically linear, and nonparametric efficient. As a consequence, $ \sqrt{n}(\psi_n^* - \psi_0) \stackrel{\sf d}{\longrightarrow} {\rm N}(0, \sigma_0^2)$ with $\sigma_0^2 \coloneq \Var_0[\varphi_0(O)]$.

\end{theorem}

Conditions \ref{cond::AIPW1} and \ref{cond::AIPW2} ensure the consistency of $\alpha_{n,j}^*$ and $\mu_{n,j}$ as estimators. Notably, Condition \ref{cond::AIPW1} only requires consistency of the best nondecreasing transformation of the initial inverse propensity score estimators $\{\pi_{n,j}^{-1}: j \in [J]\}$. On the other hand, \ref{cond::AIPW2} requires the outcome regression estimation rate to be faster than $n^{-1/6}$. Condition \ref{cond::AIPWDR} imposes a doubly robust rate condition on the nuisance estimators, only requiring a rate for the best nondecreasing transformation of the initial inverse propensity score estimators. By Theorem \ref{theorem::MSE}, \ref{cond::AIPW2} and \ref{cond::AIPWDR} together imply the standard doubly robust rate condition $\|\alpha_{n,j}^* - \alpha_0 \|_{P_0}\|\mu_{n,j} - \mu_0\|_{P_0} = o_p(n^{-1/2})$ for each $j \in [J]$. Theorem \ref{theorem::AIPW} imposes weaker conditions compared to previous literature, as it allows the initial estimator of the inverse propensity score to be incorrectly specified by an arbitrary monotone transformation. 

\section{Experiments} \label{section::simulation}

We evaluate the empirical performance of IC-IPW in semi-synthetic experiments to estimate the ATE in scenarios with near positivity violations. Specifically, we assess the performance of AIPW with weights obtained via IC-IPW against several competing methods: naive inversion, propensity score trimming, Platt scaling, and direct learning approaches. Our analysis uses semi-synthetic data from the ACIC-2017 competition \citep{hahn2019atlantic}, with covariates from the Infant Health and Development Program \citep{brooks1992effects}. Outcomes are generated from 32 distinct data-generating processes, focusing on those indexed from 17 to 24, which feature uncorrelated errors. Each process produces $M=250$ replicated datasets, each containing $n=4302$ samples.

For ATE estimation, we construct the IC-AIPW estimator using the calibrated inverse propensity weights established by Algorithm~\ref{alg:1}. We implement two variants of propensity score trimming: a deterministic method that truncates into $[0.01, 0.99]$, and an adaptive variant that learns the truncation threshold to minimize an empirical risk criterion. For indirect learning methods, we estimate cross-fitted propensity scores using gradient-boosted logistic regression in \texttt{xgboost}. The direct learning method estimates the inverse propensity score directly using \texttt{xgboost} and the loss function \( (o, \alpha) \mapsto \ind(a = a_0) \alpha(a \mid w)^2 - 2 \alpha(a_0 \mid w) \). As a parametric calibration method, we implement Platt scaling of the propensity scores as described in \citep{gutman2022propensity}. Moreover, we implement the targeted maximum likelihood estimation (TMLE) and TMLE with isotonic calibrated inverse probability weighting (IC-TMLE). Additional experimental details are given in Appendix~\ref{section::additional simulation}.

For each data-generating process, Table~\ref{tab:AIPW1} and Table~\ref{tab:AIPW2} (in Appendix~\ref{section::additional simulation}) display the Monte Carlo estimates of the bias, standard error (SE), and root-mean-square error (RMSE), as well as $95\%$ confidence interval coverage. In Table \ref{tab:AIPW1}, for data-generating processes ($18, 20, 22, 24$), where the overlap between the treatment arms is limited, we find that IC-IPW significantly reduces estimation bias and RMSE, while substantially improving the accuracy of $95\%$ coverage. For data-generating processes ($17, 19, 21, 23$), Table~\ref{tab:AIPW2} in Appendix~\ref{section::additional simulation} shows that, when overlap is less of an issue, IC-IPW performs comparably to other estimation methods regarding bias, RMSE, and coverage, with the direct learning method generally performing best.

\begin{table}[htbp]
\centering
\begin{tabular}{l|rrrr|rrrr}
\hline
\multirow{2}{*}{Method} & \multicolumn{4}{c|}{\textbf{Setting 18}} & \multicolumn{4}{c}{\textbf{Setting 20}} \\  
                        & Bias  & SE   & RMSE  & Cov.  & Bias  & SE   & RMSE  & Cov.  \\ 
\hline
Inversion               & 0.17  & 0.072 & 0.18  & 0.25  & 0.17  & 0.13  & 0.22  & 0.56  \\
Direct Learning         & 0.18  & 0.072 & 0.19  & 0.14  & 0.21  & 0.13  & 0.24  & 0.70  \\
Trimming                & 0.23  & 0.039 & 0.24  & 0.00  & 0.24  & 0.12  & 0.26  & 0.21  \\
Adaptive Trim.          & 0.17  & 0.072 & 0.18  & 0.25  & 0.17  & 0.13  & 0.22  & 0.56  \\
Platt Scaling           & 0.088 & 0.047 & 0.10  & 0.68  & 0.087 & 0.11  & \textbf{0.14} & 0.87  \\
\textbf{IC-AIPW}                 & 0.045 & 0.056 & \textbf{0.072} & \textbf{0.95}  & 0.035 & 0.17  & 0.18  & \textbf{0.92} \\
TMLE                    & 0.078 & 0.062 & 0.10  & 0.66  & 0.077 & 0.13  & 0.15  & 0.82  \\
\textbf{IC-TMLE}                 & \textbf{0.016} & 0.053 & 0.056 & 1.00  & \textbf{0.030} & 0.18  & 0.18  & \textbf{0.92} \\
Dropping                & 0.20  & 0.077 & 0.21  & 0.21  & 0.20  & 0.13  & 0.24  & 0.53  \\
\hline 
\multirow{2}{*}{Method} & \multicolumn{4}{c|}{\textbf{Setting 22}} & \multicolumn{4}{c}{\textbf{Setting 24}} \\  
                        & Bias  & SE   & RMSE  & Cov.  & Bias  & SE   & RMSE  & Cov.  \\ 
\hline
Inversion               & 0.21  & 0.087 & 0.22  & 0.18  & 0.22  & 0.13  & 0.25  & 0.53  \\
Direct Learning         & 0.13  & 0.081 & 0.14  & 0.62  & 0.18  & 0.13  & 0.22  & 0.81  \\
Trimming                & 0.27  & 0.097 & 0.28  & 0.00  & 0.29  & 0.14  & 0.30  & 0.16  \\
Adaptive Trim.          & 0.21  & 0.087 & 0.22  & 0.18  & 0.22  & 0.13  & 0.25  & 0.53  \\
Platt Scaling           & 0.14  & 0.087 & 0.14  & 0.32  & 0.13  & 0.12  & \textbf{0.17} & 0.84  \\
\textbf{IC-AIPW}                 & 0.097 & 0.076 & 0.12  & 0.79  & 0.068 & 0.21  & 0.22  & 0.92  \\
TMLE                    & 0.12  & 0.063 & 0.14  & 0.42  & 0.12  & 0.14  & 0.18  & 0.74  \\
\textbf{IC-TMLE}                 & \textbf{0.039} & 0.061 & \textbf{0.073} & \textbf{0.95}  & \textbf{0.005} & 0.22  & 0.22  & \textbf{0.93} \\
Dropping                & 0.23  & 0.076 & 0.24  & 0.17  & 0.23  & 0.14  & 0.27  & 0.52  \\
\hline
\end{tabular}
\caption{Simulation results for estimating ATE using ACIC 2017 data with data-generating processes where overlap between treatment arms is limited. Our method is in bold.}
\label{tab:AIPW1}
\end{table}

\section{Conclusion} \label{section::conclusion}

We propose a $\chi^2$-divergence measure of calibration error for inverse propensity score estimators and introduce a post-hoc algorithm that minimizes calibration error using isotonic regression with a tailored loss function for inverse propensity weights. We further emphasize the importance of developing methods that calibrate the inverse propensity score directly, as calibrating the propensity score may not yield well-calibrated inverse weights. While our primary focus is on isotonic regression for calibration, this approach can be extended to alternative methods, such as parametric scaling, histogram binning, kernel smoothing, and Venn-Abers calibration \citep{vovk2012venn, van2024self, van2025generalized}. Additionally, although we emphasize the calibration of inverse propensity weights, this method can also be applied to general inverse probability weights, including those used for handling missing data or censoring. Another interesting extension would be to study the use of the efficient plug-in learning framework \citep{van2024combining} to develop risk estimators that mitigate the poor boundary behavior of isotonic regression. 
Finally, as an adaptive histogram regressor, isotonic calibration performs automatic binning of propensity score estimates, which can be used to define subgroups with differing propensity scores for downstream analysis, such as propensity score matching \citep{xu2022isotonic}. Exploring these applications is an interesting direction for future work.


\acks{Research reported in this publication was supported by NIH grants DP2-LM013340 and  R01-HL137808, and NSF grant DMS-2210216. The content is solely the responsibility of the authors and does not necessarily represent the official views of the funding agencies.}

\bibliography{ref}

\appendix

\section{Simulation studies} \label{appendix::simulation}

\subsection{Implementation of the estimation}

Code to replicate our simulation results are available on \url{https://drive.google.com/drive/folders/1Lh4rm0vrMaaSZVnlBtgur5bvodMUzu4z}.

In our simulation studies, we estimated the nuisance parameters (i.e., $\mu_0$ and $\pi_0$) using the Super Learner \citep{van2007super}, an ensemble learning approach that utilizes cross-validation to select the optimal combination of pre-specified prediction methods. This was implemented by the \texttt{R} package \texttt{sl3} \citep{coyle2021sl3-rpkg}. For the library of prediction methods, we employed a stack of extreme gradient boosting (XGBoost) \citep{chen2016xgboost} learners with parameters maximum depth $\in \{1, 2, 3, 4, 5, 6\}$, eta $\in \{0.15, 0.2, 0.25, 0.3\}$, minimum child weight to be $5$, and max number of iterations to be $20$. To perform the isotonic calibration step, we used either the \texttt{R} function \texttt{isoreg} or a monotone 1-dimensional XGBoost with a customized loss function; both methods yielded similar results so we only report the results obtained with the first approach. 

We perform a $(0.01,0.99)$ cutoff for the fixed trimming method. For the adaptive trimming method, we employ a $(c_n,1-c_n)$ truncation with 
\begin{align*}
c_n \coloneq \argmin_{c \in [0,1/2]} \sum_{i=1}^n \alpha_{n,c}(A_i, W_i)^2 - 2 \big\{\alpha_{n,c}(1,W_i) - \alpha_{n,c}(0,W_i)\big\};
\end{align*}
and $\alpha_{n,c}(a,w) := \frac{2a-1}{\max\{\min\{\pi_n(a\mid w),  1 - c\}, c\}}$ is the truncated estimator of $\alpha_0$. 

\subsection{Performance metrics}

We estimate the performance metrics as follows. For a ATE estimator $\hat \tau$, we use $\hat \tau_m$ to denote its estimated value on the $m$-th instantiated dataset $\{o_{m,i}\}_{i=1}^n$. Also, we use $\tau_m$ to denote the true ATE of this data-generating process at each instantiated dataset, evaluated by averaging the true CATE effect values at this dataset given in the ACIC data. That is, for each $m \in M$, $\tau_m \coloneq \frac{1}{n}\sum_{i=1}^n \tau_0(w_{m,i})$, where $\tau_0$ is the true CATE function. The Monte Carlo Bias, Standard Error and RMSE are calculated as follows: 
\begin{align*}
    {\rm Bias}(\hat \tau) &\coloneq \frac{1}{M}\sum_{m=1}^M (\hat \tau_m - \tau_m), \\
    {\rm RMSE} (\hat \tau) & \coloneq \sqrt{\frac{1}{M}\sum_{m=1}^M (\hat \tau_m - \tau_m)^2}, \\
    {\rm SE} (\hat \tau) & \coloneq \sqrt{{\rm RMSE} (\hat \tau)^2 - {\rm Bias}(\hat \tau)^2};
\end{align*}
whereas the Monte Cargo $95$\% coverage is defined by:
\begin{align*}
    {\rm C} (\hat \tau) \coloneq \frac{1}{M}\sum_{m=1}^M \ind\big(\tau_m \in [\hat l_m, \hat u_m] \big),
\end{align*}
where the $\{(\hat l_m, \hat u_m)\}_{m=1}^M$'s are the estimated $95$\% confidence lower bounds and upper bounds. For CATE estimation, with a slight abuse of notations, for an arbitrary estimator $\hat \tau (w)$, we the Monte Carlo RMSE is estimated as:
\begin{align*}
    {\rm RMSE} (\hat \tau) & \coloneq \sqrt{\frac{1}{nM}\sum_{m=1}^M \sum_{i=1}^n \big(\hat \tau(w_{m,i}) - \tau_0(w_{m,i})\big)^2}.
\end{align*}
For an arbitrary estimator $\hat \alpha(w)$ of the inversed propensity weight function $\pi_0^{-1}(w)$, we estimate its RMSE by
\begin{align*}
    {\rm RMSE} (\hat \alpha) & \coloneq \sqrt{\frac{1}{nK}\sum_{k=1}^K \sum_{i=1}^n \big(\hat \alpha (w_{m,i}) - \pi_0^{-1}(w_{k,i})\big)^2},
\end{align*}
where we recall the time for iteration is $K=1000$ and evaluation is done on a out-of-fold dataset.

\subsection{Experimental results for additional settings}
\label{section::additional simulation}

\begin{table}[htbp]
\centering
\begin{tabular}{l|rrrr|rrrr}
\hline
\multirow{2}{*}{Method} & \multicolumn{4}{c|}{\textbf{Setting 17}} & \multicolumn{4}{c}{\textbf{Setting 19}} \\  
                        & Bias  & SE   & RMSE  & Cov.  & Bias  & SE   & RMSE  & Cov.  \\ 
\hline
Inversion               & 0.0068 & 0.011 & 0.013 & 0.91  & 0.0082 & 0.051 & 0.052 & 0.96  \\
Direct Learning         & \textbf{0.0030} & 0.011 & \textbf{0.011} & 0.98  & \textbf{0.0024} & 0.056 & 0.056 & 0.96  \\
Trimming                & 0.0068 & 0.011 & 0.013 & 0.91  & 0.0082 & 0.051 & 0.052 & 0.96  \\
Adaptive Trim.          & 0.0072 & 0.011 & 0.013 & 0.89  & 0.0091 & 0.051 & 0.052 & 0.96  \\
Platt Scaling           & 0.0070 & 0.011 & 0.013 & 0.91  & 0.0083 & 0.051 & 0.052 & 0.96  \\
\textbf{IC-AIPW}        & 0.0062 & 0.011 & 0.013 & 0.92  & 0.0068 & 0.052 & 0.052 & \textbf{0.95}  \\
TMLE                    & 0.0058 & 0.011 & 0.012 & \textbf{0.93}  & 0.0045 & 0.054 & 0.055 & 0.95  \\
\textbf{IC-TMLE}        & 0.0071 & 0.010 & 0.012 & 0.90  & 0.0050 & 0.053 & 0.055 & 0.96  \\
Dropping                & 0.0065 & 0.011 & 0.013 & 0.92  & 0.0060 & 0.052 & 0.053 & 0.94  \\
\hline 
\multirow{2}{*}{Method} & \multicolumn{4}{c|}{\textbf{Setting 21}} & \multicolumn{4}{c}{\textbf{Setting 23}} \\  
                        & Bias  & SE   & RMSE  & Cov.  & Bias  & SE   & RMSE  & Cov.  \\ 
\hline
Inversion               & 0.0067 & 0.013 & 0.015 & 1.00  & 0.0110 & 0.060 & \textbf{0.061} & \textbf{0.95}  \\
Direct Learning         & \textbf{0.0046} & 0.013 & \textbf{0.014} & \textbf{1.00}  & \textbf{0.0051} & 0.061 & 0.062 & 0.97  \\
Trimming                & 0.0067 & 0.013 & 0.015 & 1.00  & 0.0110 & 0.060 & 0.061 & \textbf{0.95}  \\
Adaptive Trim.          & 0.0070 & 0.013 & 0.015 & 1.00  & 0.0120 & 0.060 & 0.061 & 0.96  \\
Platt Scaling           & 0.0069 & 0.013 & 0.015 & 1.00  & 0.0120 & 0.060 & 0.061 & 0.96  \\
\textbf{IC-AIPW}        & 0.0063 & 0.013 & 0.015 & 1.00  & 0.0100 & 0.059 & 0.061 & \textbf{0.95}  \\
TMLE                    & 0.0059 & 0.013 & 0.014 & 1.00  & 0.0078 & 0.060 & 0.062 & 0.96  \\
\textbf{IC-TMLE}        & 0.0055 & 0.013 & 0.014 & 1.00  & 0.0082 & 0.060 & 0.062 & 0.95  \\
Dropping                & 0.0060 & 0.013 & 0.015 & 1.00  & 0.0090 & 0.060 & 0.062 & 0.94  \\
\hline
\end{tabular}
\caption{Simulation results for estimating ATE using ACIC 2017 data with data-generating processes where overlap between treatment arms is good. Our method is in bold.}
\label{tab:AIPW2}
\end{table}

\section{Implementation of IC-IPW}

\label{sec:code}

The following subsections provide \texttt{R} and \texttt{Python} code for constructing calibrates inverse probability weights using IC-IPW.  Here, we implement isotonic regression using \texttt{xgboost}, allowing control over the maximum tree depth and the minimum number of observations in each constant segment of the isotonic regression fit.

\subsection{R code}

\begin{verbatim}
    
# Function: isoreg_with_xgboost
# Purpose: Fits isotonic regression using XGBoost.
# Inputs:
#   - x: A vector or matrix of predictor variables.
#   - y: A vector of response variables.
#   - max_depth: Maximum depth of the trees in XGBoost (default = 15).
#   - min_child_weight: Minimum sum of instance weights (Hessian) 
#        needed in a child node (default = 20).
# Returns:
#   - A function that takes a new predictor variable x
# and returns the model's predicted values.
isoreg_with_xgboost <- function(x, y, max_depth = 15, min_child_weight = 20) {
  # Create an XGBoost DMatrix object from the data
  data <- xgboost::xgb.DMatrix(data = as.matrix(x), label = as.vector(y))
  
  # Set parameters for the monotonic XGBoost model
  params = list(max_depth = max_depth,
                min_child_weight = min_child_weight,
                monotone_constraints = 1,  # Enforce monotonic increase
                eta = 1, gamma = 0,
                lambda = 0)
  
  # Train the model with one boosting round
  iso_fit <- xgboost::xgb.train(params = params,
                                data = data, 
                                nrounds = 1)
  
  # Prediction function for new data
  fun <- function(x) {
    data_pred <- xgboost::xgb.DMatrix(data = as.matrix(x))
    pred <- predict(iso_fit, data_pred)
    return(pred)
  }
  return(fun)
}

# Function: calibrate_inverse_weights
# Purpose: Calibrates inverse weights using isotonic regression 
with XGBoost for two propensity scores.
# Inputs:
#   - A: Binary indicator variable.
#   - pi1: Cross-fitted (pooled out-of-fold) 
        propensity score estimates for treatment group A = 1.
#   - pi0: Cross-fitted (pooled out-of-fold) 
        propensity score estimates for control group A = 0.
# Returns:
#   - A list containing calibrated inverse weights for each group:
#       - alpha1_star: Calibrated inverse weights for A = 1.
#       - alpha0_star: Calibrated inverse weights for A = 0.
calibrate_inverse_weights <- function(A, pi1, pi0) {

  # Calibrate pi1 using monotonic XGBoost
  calibrator_pi1 <- isoreg_with_xgboost(pi1, A)
  pi1_star <- calibrator_pi1(pi1)
  
  # Set minimum truncation level for treated group
  c1 <- min(pi1_star[A == 1])
  pi1_star = pmax(pi1_star, c1)
  alpha1_star <- 1 / pi1_star
  
  # Calibrate pi0 using monotonic XGBoost
  calibrator_pi0 <- isoreg_with_xgboost(pi0, 1 - A)
  pi0_star <- calibrator_pi0(pi0)
  
  # Set minimum truncation level for control group
  c0 <- min(pi0_star[A == 0])
  pi0_star = pmax(pi0_star, c0)
  alpha0_star <- 1 / pi0_star
  
  # Return calibrated inverse weights for both groups
  return(list(alpha1_star = alpha1_star, alpha0_star = alpha0_star))
}


\end{verbatim}

\subsection{Python code}

\begin{verbatim}
    
import xgboost as xgb
import numpy as np

# Function: isoreg_with_xgboost
# Purpose: Fits isotonic regression using XGBoost.
# Inputs:
#   - x: A vector or matrix of predictor variables.
#   - y: A vector of response variables.
#   - max_depth: Maximum depth of the trees in XGBoost (default = 15).
#   - min_child_weight: Minimum sum of instance weights in constant segment.
# Returns:
#   - A function that takes a new predictor variable x 
and returns the model's predicted values.
def isoreg_with_xgboost(x, y, max_depth=15, min_child_weight=20):
    # Create an XGBoost DMatrix object from the data
    data = xgb.DMatrix(data=np.asarray(x), label=np.asarray(y))

    # Set parameters for the monotonic XGBoost model
    params = {
        'max_depth': max_depth,
        'min_child_weight': min_child_weight,
        'monotone_constraints': "(1)",  # Enforce monotonic increase
        'eta': 1,
        'gamma': 0,
        'lambda': 0
    }

    # Train the model with one boosting round
    iso_fit = xgb.train(params=params, dtrain=data, num_boost_round=1)

    # Prediction function for new data
    def predict_fn(x):
        data_pred = xgb.DMatrix(data=np.asarray(x))
        pred = iso_fit.predict(data_pred)
        return pred

    return predict_fn

 

def calibrate_inverse_weights(A, pi1, pi0):
    """
    Calibrates inverse weights using isotonic regression 
    with XGBoost for two propensity scores.
    
    Args:
        A (np.array): Binary indicator variable.
        pi1 (np.array): Cross-fitted (pooled out-of-fold) 
            propensity score estimates for treatment group A = 1.
        pi0 (np.array): Cross-fitted (pooled out-of-fold) 
            propensity score estimates for control group A = 0.
    
    Returns:
        dict: A dictionary containing calibrated inverse weights for each group:
              - alpha1_star: Calibrated inverse weights for A = 1.
              - alpha0_star: Calibrated inverse weights for A = 0.
    """
    
    # Calibrate pi1 using monotonic XGBoost
    calibrator_pi1 = isoreg_with_xgboost(pi1, A)
    pi1_star = calibrator_pi1(pi1)
    
    # Set minimum truncation level for treated group
    c1 = np.min(pi1_star[A == 1])
    pi1_star = np.maximum(pi1_star, c1)
    alpha1_star = 1 / pi1_star
    
    # Calibrate pi0 using monotonic XGBoost
    calibrator_pi0 = isoreg_with_xgboost(pi0, 1 - A)
    pi0_star = calibrator_pi0(pi0)
    
    # Set minimum truncation level for control group
    c0 = np.min(pi0_star[A == 0])
    pi0_star = np.maximum(pi0_star, c0)
    alpha0_star = 1 / pi0_star
    
    # Return calibrated inverse weights for both groups
    return {'alpha1_star': alpha1_star, 'alpha0_star': alpha0_star}


 \end{verbatim}

\section{Notation and Lemmas}

\subsection{Notation}
For a uniformly bounded function class $\mathcal{F}$, let $N(\epsilon,\mathcal{F},L_2(\P))$ denote the $\epsilon-$covering number \citep{van1996weak} of $\mathcal{F}$ with respect to $L_2(\P)$ and define the uniform entropy integral of $\mathcal{F}$ by  
\begin{equation*}
\mathcal{J}(\delta,\mathcal{F})\coloneq \int_{0}^{\delta} \sup_{\Q}\sqrt{\log N(\epsilon,\mathcal{F},L_2(\Q))}\,d\epsilon\ ,
\end{equation*}
where the supremum is taken over all discrete probability distributions $\Q$.

In the following, let $a_0 \in \mathcal{A}$ be arbitrary, and let $M > 0$, $C > 0$, and $\eta > 0$ be the constants associated with conditions \ref{cond::bound} and \ref{cond::positivity}. We redefine $\mathcal{F}_{\rm anti} \subset \{\theta: [-M, M] \rightarrow \mathbb{R}; \, \theta \text{ is monotone nonincreasing}\}$ to denote the family of nonincreasing functions on $\mathcal{T}$ that are uniformly bounded by $\eta^{-1}$. For a function $f: [-M, M] \rightarrow \mathbb{R}$, we denote its variation norm by $\lVert f \rVert_{TV} \coloneq \lvert f(-M) \rvert + \int_{[-M, M]} \lvert d f \rvert$. Let $\mathcal{F}_{TV}^{(a_0)}$ consist of all real-valued functions defined on $[-M, M]$ with a sup-norm bounded by $1$ and a total variation norm uniformly bounded by $3 M$, where $M$ is as described in condition \ref{cond::variation}. Additionally, define
$\mathcal{F}_{n, {\rm anti}}^{(j,a_0)} \coloneq \big\{ \theta \circ \pi_{n,j}(a_0 \given \cdot):\, \theta \in \mathcal{F}_{\rm anti} \big\}$ as the family of functions obtained by composing nonincreasing functions in $\mathcal{F}_{\rm anti}$ with $\pi_{n,j}(a_0 \mid \cdot)$. Similarly, let $\mathcal{F}_{n, {\rm TV}}^{(j,a_0)} \coloneq \{\theta \circ \pi_{n,j}(a_0 \given \cdot); \, \theta \in \mathcal{F}_{TV}^{(a_0)}\}$ represent the family of functions obtained by composing functions in $\mathcal{F}_{TV}^{(a_0)}$ with $\pi_{n,j}$. Let $\mathcal{F}_{n, {\rm Lip}}^{(j,a_0)}$ be the subset of $\big\{ (w, a) \mapsto [\ind(a = a_0) \alpha(a \mid w) - 1][\pi(a\given w) \alpha(a \mid w) - 1]: \, \pi \in \mathcal{F}_{n, {\rm TV}}^{(j,a_0)}, \, \alpha \in \mathcal{F}_{n, {\rm anti}}^{(j,a_0)} \big\}$ such that the essential supremum $\esssup_{w, a} |\pi(a \given w) \alpha(a \given w)| \leq (C + 1)$ for $C$ as given in condition \ref{cond::positivity}. Note that $\mathcal{F}_{n, {\rm Lip}}^{(j,a_0)}$ is uniformly bounded by $2(C + 1)$.

We will use the following empirical process notation: for a $\P-$measurable function $f$, we denote $\int f(o)d \P(o)$ by $\P f$. We also let $P_{n,j}$ for $j \in [J]$ denote the empirical distribution of $\mathcal{C}^{(j)}$ and, hence, write $P_{n,j}f \coloneq \frac{1}{\#| \mathcal{I}_j|}\sum_{i \in \mathcal{I}_j} f(O_i)$ with $\mathcal{I}_{j}$ indexing observations of $\mathcal{C}^{(j)} \subset \mathcal{D}_n$. For two quantities $x$ and $y$, we use the expression  $x \lesssim y$ to mean that $x$ is upper bounded by $y$ times a universal constant that may only depend on global constants that appear in Conditions \ref{cond::bound}-\ref{cond::variation}.

\subsection{Lemmas}

Recall from Algorithm \ref{alg:1} that \(\alpha_{n,j}^* = f_n^{(a_0)} \circ \pi_{n,j}\) for each \(j \in [J]\). For each \( \theta: \mathbb{R} \rightarrow \mathbb{R} \), we denote the empirical risk by:
\begin{equation*}
    R_{n}^{(a_0)}(\theta) := \sum_{i=1}^n  \Big[\ind(A_i = a_0) (\theta \circ \pi_{n,j(i)})^2(a_0 \mid W_i) - 2(\theta \circ \pi_{n,j(i)})(a_0 \mid W_i) \Big].
\end{equation*}  
Let \( \widetilde{f}_{n}^{(a_0)} := 1/g_n^{(a_0)} \) be the untruncated antitonic empirical risk minimizer in Algorithm \ref{alg:1}, and, for each $\theta \in \mathcal{F}_{anti}$, define the truncated empirical risk $ R_{n,\text{trunc}}^{(a_0)}(\theta) $ as
\begin{equation*}
   \sum_{i=1}^n  1\left\{\widetilde{f}_{n}^{(a_0)}(\pi_{n,j(i)}(a_0 \mid W_i)) < \infty \right\} \Big[\ind(A_i = a_0) (\theta \circ \pi_{n,j(i)})^2(a_0 \mid W_i) - 2(\theta \circ \pi_{n,j(i)})(a_0 \mid W_i) \Big].
\end{equation*}

\begin{lemma}[Near Risk Minimizer]
Under \ref{cond::bound}-\ref{cond::boundary}, \( {f}_{n}^{(a_0)} \) minimizes the truncated empirical risk, satisfying \( f_n^{(a_0)} \in  \arg\min_{\theta \in \mathcal{F}_{\rm anti}} R_{n,\text{trunc}}^{(a_0)}(\theta) \). Moreover, it is a near empirical risk minimizer of the untruncated empirical risk in that, for each \( C < \infty \), we have
\begin{equation*}
  R_{n}^{(a_0)}({f}_{n}^{(a_0)}) - R_{n}^{(a_0)}(\theta) \leq C O_p(n^{-2/3}),
\end{equation*}
for every \(\theta \in \mathcal{F}_{\rm anti}\) with \(\|\theta\|_{\infty} < C\).
\label{lemma::nearERM}
\end{lemma}
\begin{proof}
Since \(\mathcal{F}_{\text{iso}}\) is a convex cone, the first-order conditions of the isotonic regression problem \citep{groeneboom1993isotonic} defining \(g_{n}^{(a_0)}\) imply that \(g_{n}^{(a_0)}\) is the empirical risk minimizer if and only if, for all \(\theta \in \mathcal{F}_{\text{iso}}\), the following condition holds:
\begin{equation*}
    \frac{1}{n} \sum_{i=1}^n (\theta \circ \pi_{n,j(i)})(a_0 \mid W_i) \left\{\ind(A_i = a_0) - g_{n}^{(a_0)}(\pi_{n,j(i)}(a_0 \mid W_i)) \right\} \leq 0.
\end{equation*}
Now, for \(i \in [n]\), note that the evaluation of the isotonic regression solution \(g_{n}^{(a_0)}(\pi_{n,j(i)}(a_0 \mid W_i))\) is an empirical mean of \(\{1(A_k = a_0) : k \in [n], g_{n}^{(a_0)}(\pi_{n,j(k)}(a_0 \mid W_k)) = g_{n}^{(a_0)}(\pi_{n,j(i)}(a_0 \mid W_i))\}\) \citep{groeneboom1993isotonic}. Consequently, \(g_{n}^{(a_0)}(\pi_{n,j(i)}(a_0 \mid W_i)) = 0\) implies that \(\ind(A_i = a_0) = 0\), so the residual in the above expression is zero for boundary observations with a calibrated propensity score of \(0\). Therefore, the previous display implies that
\begin{equation*}
    \frac{1}{n} \sum_{i=1}^n  1\{g_{n}^{(a_0)}(\pi_{n,j(i)}(a_0 \mid W_i))  > 0\} (\theta \circ \pi_{n,j(i)})(a_0 \mid W_i) \left\{\ind(A_i = a_0) - g_{n}^{(a_0)}(\pi_{n,j(i)}(a_0 \mid W_i)) - 1\right\} \leq 0.
\end{equation*}
Noting that \(\widetilde{f}_{n}^{(a_0)} = 1/g_{n}^{(a_0)}\) and $\theta$ is arbitrary, we conclude that \(\widetilde{f}_{n}^{(a_0)}\) satisfies, for all \(\theta \in \mathcal{F}_{\text{iso}}\), the following condition:
\begin{equation*}
    \frac{1}{n} \sum_{i=1}^n  1\{\widetilde{f}_{n}^{(a_0)}(\pi_{n,j(i)}(a_0 \mid W_i)) < \infty\} (\theta \circ \pi_{n,j(i)})(a_0 \mid W_i) \left\{\ind(A_i = a_0)\widetilde{f}_{n}^{(a_0)}(\pi_{n,j(i)}(a_0 \mid W_i)) - 1\right\} \leq 0.
\end{equation*}

In view of these first-order conditions, it follows that $\widetilde{f}_{n}^{(a_0)} $ minimizes the empirical risk restricted to observations with finite weights, that is,
\begin{equation*}
 \widetilde{f}_{n}^{(a_0)}  \in \argmin_{\theta \in \mathcal{F}_{\rm anti} } R_{n,trunc}^{(a_0)}(\theta) .
 \end{equation*}  
Hence, $R_{n,trunc}^{(a_0)}(\widetilde{f}_{n}^{(a_0)}) - R_{n,trunc}^{(a_0)}(\theta) \leq 0$ for all $\theta \in \mathcal{F}_{anti}$.
Furthermore, since \({f}_{n}^{(a_0)} = \widetilde{f}_{n}^{(a_0)} \wedge b_n^{(a_0)}\) where \(b_n^{(a_0)} = \max_{i \in [n]: A_i = a_0} f_n^{(a_0)}(\pi_{n,j(i)}(a_0 \mid W_i))\) is the maximum finite value of $\widetilde{f}_{n}^{(a_0)}$, we know that \({f}_{n}^{(a_0)}(\pi_{n,j(i)}(a_0 \mid W_i)) = \widetilde{f}_{n}^{(a_0)}(\pi_{n,j(i)}(a_0 \mid W_i))\) for all \(i\) such that \(\widetilde{f}_{n}^{(a_0)}(\pi_{n,j(i)}(a_0 \mid W_i)) < \infty\). 
As a consequence, the same inequality holds for the truncated minimizer \({f}_{n}^{(a_0)}\):
\begin{equation*}
   R_{n,trunc}^{(a_0)}({f}_{n}^{(a_0)}) - R_{n,trunc}^{(a_0)}(\theta) \leq 0 \text{ for all } \theta \in \mathcal{F}_{anti}.
\end{equation*}
However, we know that infinite weight estimates can only occur at the boundary observations \(\#\{i \in [n]: A_i \neq a_0 \text{ and } \pi_{n,j(i)}(a_0 \mid W_{i}) < \min_{k \in [n]: A_k = a_0} \pi_{n,j(i)}(a_0 \mid W_{k})\}   = 1/b_n^{(a_0)}\) and, by \ref{cond::boundary}, there are at most \(O_p(n^{1/3})\) such observations. By \ref{cond::positivity}, \({f}_{n}^{(a_0)}(\pi_{n,j(i)}(a_0 \mid W_i))\) is almost surely uniformly bounded for each \(i \in [n]\), and, therefore,
\begin{equation*}
  R_{n}^{(a_0)}({f}_{n}^{(a_0)}) - R_{n}^{(a_0)}(\theta) \leq C O_p(n^{1/3}/n) = CO_p(n^{-2/3}),
\end{equation*}
for all \(M < \infty\) and \(\theta \in \mathcal{F}_{\rm anti}\) with \(\|\theta\|_{\infty} < C\).

\end{proof}

\begin{lemma}[Empirical Calibration]
\label{lem:lem1}
For calibrated predictors \(\alpha_{n,j}^*(a_0 \mid \cdot) = f_n^{(a_0)} \circ \pi_{n,j}(a_0 \mid \cdot)\) indexed by \(j \in [J]\) obtained using Algorithm \ref{alg:1}, we have
\begin{equation}
\label{eq5:cal}
  \frac{1}{n} \sum_{i=1}^n \big[h \circ \alpha_{n,j(i)}^*(a_0 \mid W_i)\big] \big\{1(A_i = a_0)\alpha_{n,j(i)}^*(a_0 \mid W_i) - 1\big\} = M O_p(n^{-2/3}).
\end{equation}
for each \(M < \infty\) and every \(h \in \mathcal{F}_{\rm anti}\) with \(\|h\|_{\infty} < M\).
\end{lemma}
\begin{proof}
This proof follows from a modification of the argument used to establish Lemma C.1 in \cite{van2023causal}. For each $g:\mathbb{R} \rightarrow \mathbb{R}$, denote the empirical least squares risk by
    $$R_n(g) := \sum_{i=1}^n  \big\{\ind(A_i = a_0) - (g \circ \pi_{n, j(i)})(a_0 \given W_i) \big\}^2.$$
Recall that $f_n^{(a_0)} = 1 / (c_n^{(a_0)} \vee g_n^{(a_0)})$, where $g_n^{(a_0)}$ is obtained via 
    \begin{align*}
        g_n^{(a_0)} \in  \argmin_{g \in \mathcal{F}_{\rm iso}} \sum_{i=1}^n  \big\{\ind(A_i = a_0) - (g \circ \pi_{n, j(i)})(a_0 \given W_i) \big\}^2 = \argmin_{g \in \mathcal{F}_{\rm iso}} R_n(g).
    \end{align*}
The isotonic regression solution $(g_n^{(a_0)}$ can be expressed as a piecewise-constant function, satisfying
    \begin{align*}
        (g_n^{(a_0)} \circ \pi_{n, j})(a_0 \given w) = b_0 + \sum_{t=1}^T b_t \ind(\pi_{n, j}(a_0 \given w) \geq u_t), \quad \text{for each } j \in [J],
    \end{align*}
with $T+1 \in \mathbb{N}$ values determined by the coefficients $\{b_t\}_{t=0}^T$ and thresholds $\{u_t\}_{t=1}^T$ \citep{barlow1972isotonic}. By monotonicity of $g_n^{(a_0)}$, it must hold that $b_0 \in \mathbb{R}$ and $b_t > 0$ for all $t \in \{1, \dots, T\}$. 

For any jump point $u_t$ of $g_n^{(a_0)}$, we can always choose $\varepsilon \in \mathbb{R}$ with a small enough absolute value such that 
    \begin{align*}
         \xi_n(\varepsilon) &: x \mapsto g_n^{(a_0)}(x)  + \varepsilon \ind(x \geq u_t), \quad \text{for each } j \in [J],
    \end{align*}
is an isotonic function in $\mathcal{F}_{iso}$. In particular, this holds for any $\varepsilon$ such that $|\varepsilon| < \min \{b_t\}_{t=1}^T$.

Since \(g_n^{(a_0)}\) is the minimizing isotonic solution and \(R_n(\xi_n(0)) = R_n(g_n^{(a_0)})\), we have that \(R_n(\xi_n(\varepsilon)) \geq R_n(g_n^{(a_0)})\) for all \(\varepsilon\) sufficiently close to zero. Thus, since \(\varepsilon \mapsto R_n(\xi_n(\varepsilon))\) is differentiable with a local minimum achieved at \(\varepsilon = 0\), the derivative \(\frac{d}{d\varepsilon} R_n(\xi_n(\varepsilon)) \big |_{\varepsilon = 0}\) is zero. Computing the derivative, we find that
\begin{align*}
    \sum_{i=1}^n \ind( \pi_{n, j(i)}(a_0 \mid W_i) \geq u_t) \big\{\ind(A_i = a_0) - (g_n^{(a_0)} \circ \pi_{n, j(i)})(a_0 \mid W_i) \big\} = 0.
\end{align*}
Since the jump point \(u_t\) can be chosen arbitrarily, we can show that the same equation holds for linear combinations of indicator functions of the form \(a_{n, j}(w) = \sum_{t=0}^T b_t' \ind(\pi_{n, j}(a_0 \mid w) \geq u_t)\). That is,
\begin{align*}
    \sum_{i=1}^n a_{n, j(i)}(W_i) \big\{\ind(A_i = a_0) - (g_n^{(a_0)} \circ \pi_{n, j(i)})(a_0 \mid W_i) \big\} = 0.
\end{align*}

Let $\widetilde{\alpha}_{n, j}^*(a_0 \given \cdot) := 1/ (g_n^{(a_0)} \circ \pi_{n, j})(a_0 \given \cdot) $ denote the untruncated isotonic calibrated weights. Then, we have
     \begin{align*}
    \sum_{i=1}^n a_{n, j(i)}(W_i) \big\{\ind(A_i = a_0) -  \{\widetilde{\alpha}_{n, j(i)}^*(a_0 \given W_i) \}^{-1} \big\} &= 0.
    \end{align*}
Next, observe that, for any $h:\mathbb{R} \cup \{\infty\} \rightarrow \mathbb{R}$, the function $h \circ \widetilde{\alpha}_{n, j}^*$ must take the form of $a_{n, j} $, since a transformation of a step function preserves its jump points. Hence, we have with probability tending to one,
    \begin{align*}
     \sum_{i=1}^n \big[h \circ \widetilde{\alpha}_{n, j(i)}^*(a_0 \given W_i)\big]  \big\{\ind(A_i = a_0) -  \{\widetilde{\alpha}_{n, j(i)}^*(a_0 \given W_i) \}^{-1} \big\} &= 0.
    \end{align*}
In addition, since \(h: \mathbb{R} \cup \{\infty\} \rightarrow \mathbb{R}\) is arbitrary, we have, by Condition \ref{cond::positivity}, with probability tending to one, for any \(h:\mathbb{R} \cup \{\infty\} \rightarrow \mathbb{R}\) with \(h(\infty) := 0\), that
\begin{align*}
     \sum_{i=1}^n  \big[h \circ \widetilde{\alpha}_{n, j(i)}^*(a_0 \mid W_i)\big]   \left\{ \ind(A_i = a_0) \widetilde{\alpha}_{n, j(i)}^*(a_0 \mid W_i) - 1 \right\} &= 0.
\end{align*}
Next, observe, by definition of the truncation in Algorithm \ref{alg:1}, that $\widetilde{\alpha}_{n, j(i)}^*(a_0 \given W_i)$ can only differ from $\alpha_{n, j(i)}^*(a_0 \given W_i)$ for $i \in [n]$ such that $A_i \neq a_0$ and $g_n^{(a_0)}(\pi_n(a_0 \mid W_{i})) = 0$. Thus, by \ref{cond::boundary}, the two functions can only differ at most at $O_p(n^{1/3})$ observations. It follows from \ref{cond::positivity} that
\begin{align*}
     \sum_{i=1}^n  \big[h \circ {\alpha}_{n, j(i)}^*(a_0 \given W_i)\big]   \left\{ \ind(A_i = a_0){\alpha}_{n, j(i)}^*(a_0 \given W_i) - 1 \right\} &=M O_p(n^{1/3}),
    \end{align*}
and, therefore,
\begin{align*}
    \frac{1}{n} \sum_{i=1}^n  \big[h \circ {\alpha}_{n, j(i)}^*(a_0 \given W_i)\big]   \left\{\ind(A_i = a_0)\alpha_{n, j(i)}^*(a_0 \given W_i) - 1 \right\} &= MO_p(n^{-2/3}),
    \end{align*}

\end{proof}

\begin{lemma}
Under Condition \ref{cond::variation}, for all \(j \in [J]\), the map \(g_{0,n,j}^{*(a_0)}\), defined such that \(g_{0,n,j}^{*(a_0)}(\pi_{n,j}(a_0 \mid W)) = \gamma_0^{(a_0)}(W, \alpha_{n,j}^*)\), has its total variation, \(P_0\)-almost surely, bounded above by \(3 M\), where \(M\) is as specified in Condition \ref{cond::variation}.
\label{lemma::TVnormBounded}
\end{lemma}
\begin{proof}    
This proof follows from a modification of the argument used to establish Lemma C.3 in \cite{van2023causal}. Recall that $\alpha_{n,j}^{*-1}(a_0 \mid w) = \{g_n^{(a_0)} \circ \pi_{n,j}(a_0 \mid w)\}^{-1}$, where $g_n^{(a_0)}$ is an isotonic step function defined in Algorithm \ref{alg:1}. Therefore, $\gamma_0^{(a_0)}(\cdot, \alpha_{n,j}^{*-1}) = \gamma_0^{(a_0)}(\cdot, g_n^{(a_0)} \circ \pi_{n,j}(a_0 \mid \cdot))$.


Note that the map $g_{0,n,j}^{*(a_0)}$ such that $g_{0,n,j}^{*(a_0)}(\pi_{n,j}(a_0\given w))= \gamma_0^{(a_0)}(w, \alpha_{n,j}^*)$ is given by
$$t \mapsto  E_0 \big[ \pi_0(a_0 \given W) \given g_n^{(a_0)} \circ \pi_{n,j}(a_0 \mid W) = g_n^{(a_0)}(t), \mathcal{D}_n\big].$$
Since $g_n^{(a_0)}$ is nondecreasing and piecewise constant, we have, for any $t \in \mathbb{R}$, that
\begin{align*}
    E_0 \big[ \pi_0(a_0 \given W) \given g_n^{(a_0)} \circ \pi_{n,j}(a_0 \mid W) = g_n^{(a_0)}(t), \mathcal{D}_n\big] =  E_0 \big[ \pi_0(a_0 \given W) \given \pi_{n,j}(a_0 \given W) \in B_{t}, \mathcal{D}_n \big]
\end{align*}
for the set $B_{t} := \left\{ z \in \mathbb{R}: g_n^{(a_0)}(z) = g_n^{(a_0)}(t)\right\}$, where $B_{t} = \{z \in \mathbb{R}: a(t) \leq z < b(t)\}$ for some endpoints $a(t), b(t) \in \mathbb{R}$. The law of total expectation further implies that
$$E_0[\pi_0(a_0 \given W)  \,|\, \pi_{n,j}(a_0 \given W) \in B_{t}] = E_0[g_{0,n,j}^{(a_0)} \circ \pi_{n,j}(a_0 \given W) \,|\, \pi_{n,j}(a_0 \given W) \in B_{t}]\ ,$$
where $g_{0,n,j}^{(a_0)} $ in \ref{cond::variation} is such that $g_{0,n,j}^{(a_0)} \circ \pi_{n,j}(a_0 \given W) = \gamma_0(W, \pi_{n,j}(a_0 \given W))$ $P_0$-almost surely. By Condition \ref{cond::variation}, the function $g_{0,n,j}^{(a_0)}$ is of bounded total variation. Heuristically, since $t \mapsto E[g_{0,n,j}^{(a_0)} \circ \pi_{n,j}(a_0 \given W) \,|\, \pi_{n,j}(a_0 \given W)] \in B_{t}$ is obtained by locally averaging $g_{0,n,j}^{(a_0)}$ within the bins $(B_{t}: t)$, its total variation should also be bounded. Up to notation, this proof follows from the proof of Lemma C.3 in \cite{van2023causal}. In particular, their proof establishes that $t \mapsto E[g_{0,n,j}^{(a_0)} \circ \pi_{n,j}(a_0 \given W) \,|\, \pi_{n,j}(a_0 \given W) \in B_{t}]$ and, therefore also $g_{0,n,j}^{*(a_0)}$, has total variation norm bounded above by $3 \|g_{0,n,j}^{(a_0)}\|_{TV}$ almost surely. The result then follows, noting that $\lVert g_{0,n,j}^{(a_0)} \rVert_{TV} \leq M$ by \ref{cond::variation}.

\end{proof}

\section{Proofs of Theorems}
\subsection{Proof of Theorem \ref{theorem::CAL}}

For each $j \in [J]$ and any function $h: \mathbb{R} \rightarrow \mathbb{R}$, we have, by the law of iterated expectations, that
\begin{align*}
 &    \int \big[h \circ \alpha_{n,j}^*(a_0 \given w) \big] \big\{\ind(A=a_0)\alpha_{n,j}^*(a_0 \given w) - 1\big\} dP_0(w,a) \\
 & ~~ = \int \big[h \circ \alpha_{n,j}^*(a_0 \given w) \big] \big\{ \gamma_{0}^{(a_0)}(w, \alpha_{n,j}^{*-1})\alpha_{n,j}^*(a_0 \given w) - 1\big\} dP_0(w).   
\end{align*}
Choosing $h$ such that $h(\alpha_{n,j}^*(w)) = \gamma_0^{(a_0)}(w, \alpha_{n,j}^{*-1})\alpha_{n,j}^*(w) - 1$ in the above display, we find that
\begin{align}
    & \int \big\{\gamma_{0}^{(a_0)}(w, \alpha_{n,j}^{*-1})\alpha_{n,j}^*(a_0 \given w) - 1 \big\}\big\{\ind(A=a_0)\alpha_{n,j}^*(a_0 \given w) - 1\big\} dP_0(w,a) \nonumber \\
    & \qquad = \int   \big\{ \gamma_{0}^{(a_0)}(w, \alpha_{n,j}^{*-1})\alpha_{n,j}^*(a_0 \given w) - 1\big\}^2 dP_0(w).
    \label{proof::eqn1}
\end{align}
By Lemma \ref{lem:lem1}, we also have 
$$\frac{1}{J}\sum_{j \in [J]} \sum_{i \in \mathcal{C}^{(j)}}\big\{\gamma_0^{(a_0}(W_i, \alpha_{n,j}^{*-1})\alpha_{n,j}^*(a_0 \given W_i) - 1\big\} \big\{1(A_i = a_0)\alpha_{n,j}^*(a_0 \given W_i) - 1\big\} = O_p(n^{-2/3}). $$
Combining this with Equation \eqref{proof::eqn1}, we obtain the identity
\begin{align*}
&\frac{1}{J}\sum_{j \in [J]} {\rm CAL}^{(a_0)}(\alpha_{n,j}^*) \\
& ~= \frac{1}{J}\sum_{j \in [J]}  \int \big\{ \gamma_{0}^{(a_0)}(w, \alpha_{n,j}^{*-1})\alpha_{n,j}^*(a_0 \given w) - 1\big\}^2 dP_0(w) \nonumber \\
 & ~= \frac{1}{J}\sum_{j \in [J]} \int \left[ f_{n,j}^{(a_0)}(w,a)\left\{\gamma_{0}^{(a_0)}(w, \alpha_{n,j}^{*-1})\alpha_{n,j}^*(a_0 \given w) - 1\right\} \right] \d(P_0-P_{n,j})(w, a) + O_p(n^{-2/3}) \ , 
\yestag \label{proof::eqn2}
\end{align*}
where $(w,a) \mapsto f_{n,j}^{(a_0)}(w,a) \coloneq \{\ind(a=a_0)\alpha_{ n,j}^*(a \given w) - 1 \}  $.

We now use empirical process techniques to obtain a convergence rate for the left-hand side of the above display. Our proof follows along the lines of the proof of Theorem 3 in \cite{van2023causal}. To this end, by definition of $f_n^{(a_0)}$ and \ref{cond::positivity}, $w \mapsto \alpha_{n,j}^*(a_0 \given w) = f_n^{(a_0)} \circ \pi_{n,j}(a_0 \given w)$ is a bounded monotone transformation of $w \mapsto \pi_{n,j}(a_0 \given w)$ and falls in the random function class $\mathcal{F}_{n, {\rm anti}}^{(j,a_0)}$. Under \ref{cond::variation} and Lemma \ref{lemma::TVnormBounded}, $\pi_{n,j}(w) \mapsto \gamma_{0}^{(a_0)}(w, \alpha_{n,j}^{*-1})$ has total variation norm bounded by $3K$ for some constant $K > 0$. Hence, $\gamma_{0}^{(a_0)}(w, \alpha_{n,j}^*)$ falls almost surely in the random function class $\mathcal{F}_{n, {\rm TV}}^{(j,a_0)}$. We conclude, by definition, that $ (w,a) \mapsto f_{n,j}^{(a_0)}(w,a)\{\gamma_{0}^{(a_0)}(w, \alpha_{n,j}^{*-1})\alpha_n^*(a_0 \given w) - 1\}$ falls in the function class $\mathcal{F}_{n,{\rm Lip}}^{(j,a_0)}$.  

Recall the positivity constant \(\eta > 0\) from Condition \ref{cond::positivity} and note that \(\esssup_{a,w}\lvert f_{n,j}^{(a_0)}(w)\rvert \leq \eta + 1\). Denote the localized function class $\mathcal{F}_{n, j}(\delta_n) = \{f \in \mathcal{F}_{n,{\rm Lip}}^{(j,a_0)} : \|f\| \leq J(\eta + 1) \delta_n\}$. Using Hölder's inequality and Jensen's inequality, Equation \eqref{proof::eqn2} implies
\begin{align}
    \delta_n^2 &\leq \frac{1}{J} \sum_{j \in [J]} \sup_{f \in \mathcal{F}_{n, j}(\delta_n)} (P_0 - P_{n,j}) f + O_p(n^{-2/3}).
    \label{proof::eqn3}
\end{align}
where  \(\delta_n^2 \coloneq J^{-1} \sum_{j=1}^J \int \big\{ \gamma_{0}^{(a_0)}(w, \alpha_{n,j}^{*-1})\alpha_{n,j}^*(a_0 \mid w) - 1 \big\}^2 dP_0(w)\) is the fold-averaged calibration error. 
We denote the first term on the right-hand side of \eqref{proof::eqn3} by \(\frac{1}{J} \sum_{j \in [J]} \phi_{n,j}(\delta_n)\), where
\[
    \phi_{n,j}(\delta) \coloneq \sup_{f \in \mathcal{F}_{n, j}(\delta)} (P_0 - P_{n,j}) f.
\]
and define the random variable 
\[
Z_n :=  \delta_n^2 - \frac{1}{J} \sum_{j \in [J]} \sup_{f \in \mathcal{F}_{n, j}(\delta_n)} (P_0 - P_{n,j}) f,
\]
where we note $Z_n = O_p(n^{-2/3})$ by \eqref{proof::eqn3}.

Showing the asserted stochastic order, $\delta_n^2 = O_p(n^{-2/3})$, is equivalent to showing for all $\epsilon>0$, we can find a $2^S$ sufficiently large so that $\limsup_{n \to\infty} \P(n^{2/3} \delta_n^2 >2^S)<\epsilon$. To this end, we just need to show $\lim_{n\to\infty} \P(n^{2/3} \delta_n^2 >2^S)\to 0$ as $S \to \infty$. Let $K_{\varepsilon} < \infty$ be such that $Z_n \leq K_{\varepsilon} n^{-2/3}$ with probability at least $1 - \varepsilon$. Define the event $A_s \coloneq \left\{ n^{2/3} \delta_n^2 \in (2^s, 2^{s+1}]\right\}$ for each $s$. From a peeling argument and Markov's inequality, we find
\begin{align*} 
    \P\big(n^{2/3} \delta_n^2 >2^S \big) &\leq \P\big(n^{2/3} \delta_n^2 >2^S, Z_n \leq  K_{\varepsilon}n^{-2/3} \big) + \varepsilon \\
  &\leq \sum_{s=S}^\infty \P \big(2^{s+1} \geq n^{2/3} \delta_n^2 > 2^s ,  Z_n \leq  K_{\varepsilon}n^{-2/3}\big)  + \varepsilon \\
    & \leq \sum_{s=S}^\infty \P \Big(A_s, \delta_n^2 \leq \frac{1}{J}\sum_{j \in [J]} \phi_{n,j}(\delta_n) + Z_n, Z_n \leq  K_{\varepsilon}n^{-2/3} \Big) + \varepsilon \\
    & \leq \sum_{s=S}^\infty \P \Big(A_s, \delta_n^2 \leq \frac{1}{J}\sum_{j \in [J]} \phi_{n,j}(\delta_n)+ K_{\varepsilon}n^{-2/3}  \Big) + \varepsilon \\
    & \leq \sum_{s=S}^\infty \P \Big( n^{-2/3} 2^s < \delta_n^2 \leq \frac{1}{J}\sum_{j \in [J]} \phi_{n,j} \big(n^{-1/3} 2^{\frac{s+1}{2}}\big)  + K_{\varepsilon}n^{-2/3} \Big) + \varepsilon \\
    & \leq \sum_{s=S}^\infty \P \Big( n^{-2/3} 2^s < \frac{1}{J}\sum_{j \in [J]} \phi_{n,j}\big(n^{-1/3} 2^{\frac{s+1}{2}}\big) 
 + K_{\varepsilon}n^{-2/3} \Big) + \varepsilon  \\
    & \leq \sum_{s=S}^\infty \frac{J^{-1} \sum_{j \in [J]}E_0^n\big[ \phi_{n,j}\big(n^{-1/3} 2^{\frac{s+1}{2}}\big)\big] + K_{\varepsilon}n^{-2/3}  }{n^{-2/3} 2^s} + \varepsilon.
\end{align*}
In view of the above display, we aim to bound the local modulus of continuity, $E_0^n[\phi_{n,j}(\delta)]$ for a fixed $\delta > 0$. Using that $\mathcal{F}_{n,{\rm Lip}}^{(j,a_0)}$ is uniformly bounded and applying Theorem 2.1 in \cite{van2011local} conditional on $\mathcal{E}^{(j)}$, we find
$$E_0^n[\phi_{n,j}(\delta) \given \mathcal{E}^{(j)}] \lesssim  n^{-1/2}   \mathcal{J}\big(\delta, \mathcal{F}_{n,{\rm Lip}}^{(j,a_0)}\big) \Big(1 + \frac{\mathcal{J}\big(\delta, \mathcal{F}_{n,{\rm Lip}}^{(j,a_0)}\big)}{\sqrt{n} \delta^2} \Big) .$$

We note that, conditional on \(\mathcal{E}^{(j)}\), \(\mathcal{F}_{n,{\rm Lip}}^{(j,a_0)}\) is a multivariate Lipschitz transformation of \(\mathcal{F}_{n, {\rm TV}}^{(j,a_0)}\) and \(\mathcal{F}_{n, {\rm anti}}^{(j,a_0)}\). Therefore, by Theorem 2.10.20 of \cite{van1996weak}, we have that \(\mathcal{J}(\delta, \mathcal{F}_{n,{\rm Lip}}^{(j,a_0)}) \lesssim \mathcal{J}(\delta, \mathcal{F}_{n, {\rm TV}}^{(j,a_0)}) + \mathcal{J}(\delta, \mathcal{F}_{n, {\rm anti}}^{(j,a_0)})\). Since functions of bounded total variation can be written as a difference of non-decreasing monotone functions, we have by the same theorem that \(\mathcal{J}(\delta, \mathcal{F}_{\rm TV}^{(a_0)}) \lesssim \mathcal{J}(\delta, \mathcal{F}_{\rm anti})\). By the covering number bound for bounded monotone functions given in Theorem 2.7.5 of \citet{van1996weak}, we have \(\mathcal{J}(\delta, \mathcal{F}_{\rm anti}) \lesssim \sqrt{\delta}\). We have by the change-of-variables formula
\begin{align*}
\mathcal{J}(\delta, \mathcal{F}_{n, {\rm anti}}^{(j,a_0)}) & =  \int_0^{\delta} \sup_\Q \sqrt{N\big(\varepsilon, \mathcal{F}_{n, {\rm anti}}^{(j,a_0)}, \norm{\,\cdot\,}_\Q \big)}\,\d\varepsilon =  \int_0^{\delta} \sup_\Q \sqrt{N(\varepsilon, \mathcal{F}_{\rm anti}, \norm{\,\cdot\,}_{Q \circ \pi_{n,j}^{-1}} )}\, \d\varepsilon  = \mathcal{J}(\delta, \mathcal{F}_{\rm anti}), \\
\mathcal{J}(\delta, \mathcal{F}_{n, {\rm TV}}^{(j,a_0)}) & =  \int_0^{\delta} \sup_\Q \sqrt{N\big(\varepsilon, \mathcal{F}_{n, {\rm TV}}^{(j,a_0)}, \norm{\,\cdot\,}_\Q\big)}\,\d\varepsilon =  \int_0^{\delta} \sup_\Q \sqrt{N(\varepsilon, {F}_{n, {\rm TV}}, \norm{\,\cdot\,}_{Q \circ \pi_{n,j}^{-1}})}\, \d\varepsilon  = \mathcal{J}(\delta, \mathcal{F}_{n, {\rm TV}});
\end{align*}
where, with a slight abuse of notation, \(Q \circ \pi_{n,j}^{-1}\) is the push-forward probability measure for the random variable \(\pi_{n,j}(a_0 \mid W)\) conditional on \(\mathcal{E}^{(j)}\).

We are now in a position to apply Theorem 3.2.5 of \citet{van1996weak} to obtain a rate of convergence for the fold-averaged calibration error $\delta_n^2$. Returning to our previous bound for $E_0^n[\phi_{n,j}(\delta) \given \mathcal{E}^{(j)}]$ and applying the law of total expectation, we conclude that
\[
\frac{1}{J} \sum_{j=1}^J E_0^n[\phi_{n,j}(\delta)] \lesssim  n^{-1/2} \mathcal{J}(\delta, \mathcal{F}_{\text{anti}} ) \Bigg(1 + \frac{\mathcal{J}(\delta, \mathcal{F}_{\text{anti}})}{\sqrt{n} \delta^2} \Bigg),
\]
where the right-hand side is non-random. This gives us
\begin{align*}
    \frac{J^{-1} \sum_{j \in [J]}E_0^n\big[ \phi_{n,j}\big(n^{-1/3} 2^{\frac{s+1}{2}}\big)\big]}{n^{-2/3} 2^s} &\lesssim 2^{-3s/4}
\end{align*}
since $\mathcal{J}(\delta, \mathcal{F}_{\text{anti}}) \lesssim \sqrt{\delta}$. Our previous display implies that for any $n \in \mathbb{N}$,
\begin{align*}
   P \big(n^{2/3} \delta_n^2 >2^S \big) &\lesssim \varepsilon + \sum_{s=S}^\infty \frac{1}{2^{3s/4}} \overset{S \to \infty}{\to} \varepsilon.
\end{align*}
Since $\varepsilon > 0$ was arbitrary, we conclude that $\mathbb{P} \big(n^{2/3} \delta_n^2 >2^S \big)  \overset{S \to \infty}{\to} 0$.
Taking the limit on both sides with respect to $n$ gives us the desired result:
\[
\delta_n^2 =  \frac{1}{J}\sum_{j \in [J]}  \int   \big\{ \gamma_{0}^{(a)}(w, \alpha_{n,j}^*)\alpha_{n,j}^*(a \given w) - 1\big\}^2 dP_0(w) = O_p(n^{-2/3}).
\]

\subsection{Proof of Theorem \ref{theorem::MSE}}

The population risk function used to define the oracle transformations $\{f_{n,0}^{(a_0)}: a_0 \in \mathcal{A}\}$ satisfies
\begin{align*}
&\sum_{j \in [J]} \int  \big\{  f^{(a)} \circ \pi_{n,j}(a \given w)  -   \pi_0^{-1}(a \given w) \big\}^2  dP_0(a,w) \\
& \quad = \sum_{a_0 \in \mathcal{A}} \sum_{j \in [J]} \int \ind(A=a_0)  \big\{  f^{(a_0)} \circ \pi_{n,j}(a_0 \given w)  -   \pi_0^{-1}(a_0 \given w) \big\}^2  dP_0(a,w).    
\end{align*}
Since the oracle transformations are variation independent, we have for $a_0 \in \mathcal{A}$,
\begin{align*}
    f_{n,0}^{(a_0)} &\in \argmin_{f \in \mathcal{F}_{\rm anti}}   \sum_{j \in [J]}  \int \ind(A=a_0)  \big\{  f \circ \pi_{n,j}(a_0 \given w)  -   \pi_0^{-1}(a_0 \given w) \big\}^2  dP_0(a,w) \\
    &= \argmin_{f \in \mathcal{F}_{\rm anti}} \sum_{j \in [J]} \int \ell_{n,j}^{(a_0)}(a_0,w; f) dP_0(a,w) = \argmin_{f \in \mathcal{F}_{\rm anti}} R^{(a_0)}(f),
\end{align*} 
where, for each $j \in [J]$, we define the fold-specific loss function and the population risk
\begin{align*}
    (a, w, f) \mapsto  \ell_{n,j}^{(a_0)}(a,w; f) & \coloneq  \left\{ \ind(A=a_0)  f \circ \pi_{n,j}(a_0 \given w)^2 - 2  f \circ \pi_{n,j}(a_0 \given w) \right\}; \\
    f \mapsto R^{(a_0)}(f) & \coloneq \sum_{j \in [J]} \int \ell_{n,j}^{(a_0)}(a_0,w; f) dP_0(a,w).
\end{align*}

We claim that Condition \ref{cond::positivity2} implies \(\esssup_{a,w} |f_{n,0}^{(a)} \circ \pi_{n,j}(a \given w)| < \eta^{-1}\). This holds because, otherwise, we could truncate \(f_{n,0}^{(a)}\) to lie in \([-1/\eta, 1/\eta]\) almost surely and obtain a monotone non-increasing function with smaller population risk. To see this, note that, for any $f \in \mathcal{F}_{anti}$,
\begin{align*}
    & \sum_{j \in [J]}  \int \ell_{n,j}^{(a_0)}(a_0, w; f) \, dP_0(a, w) \\
    & = \sum_{j \in [J]} \int \left\{ 1(|f| \leq \eta^{-1}) + 1(|f| > \eta^{-1}) \right\} \ell_{n,j}^{(a_0)}(a_0, w; f) \, dP_0(a, w) \\
    & = \sum_{j \in [J]}  \int \ind(A = a_0) \, 1(|f| \leq \eta^{-1}) \left\{ f \circ \pi_{n,j}(a \given w) - \pi_0^{-1}(a_0 \given w) \right\}^2 \, dP_0(a, w) \\
    & \quad + \sum_{j \in [J]}  \int \ind(A = a_0) \, 1(|f| > \eta^{-1}) \left\{ f \circ \pi_{n,j}(a \given w) - \pi_0^{-1}(a_0 \given w) \right\}^2 \, dP_0(a, w),
\end{align*}
where the global minimizer over \(f\) of the left-hand side also minimizes the first term on the right-hand side by \ref{cond::positivity2}, and the second term on the right-hand side also takes its minimal value (zero) for any function \(f\) taking values in \([-1/\eta, 1/\eta]\). Thus, we also have \(f_{n,0}^{(a_0)} \in \mathcal{F}_{\text{anti}}^{1/\eta}\).

For $a_0 \in \mathcal{A}$, denote $(w,a_0) \mapsto \alpha_{0,j}^*(a_0 \given w) := f_{n,0}^{(a_0)} \circ \pi_{n,j}(a_0 \given w)$. By the law of iterated expectations, the excess risk satisfies for each $j \in [J]$:
\begin{align}
    & P_0\ell_{n,j}^{(a_0)}(f_n^{(a_0)}) - P_0\ell_{n,j}^{(a_0)}(f_{n,0}^{(a_0)}) \\
    & =\int \left[\pi_0(a_0 \given w)\left\{(\alpha_{n,j}^*(a_0 \given w))^2 - (\alpha_{0,j}^*(a_0 \given w))^2\right\} - 2 \left\{\alpha_{n,j}^*(a_0 \given w) - \alpha_{0,j}^*(a_0 \given w) \right\}\right]  dP_0(w) \nonumber \\
    & = \int \left[ \alpha_{n,j}^*(a_0 \given w) - \alpha_{0,j}^*(a_0 \given w) \right] \left[ \pi_0(a_0 \given w)\left\{\alpha_{n,j}^*(a_0 \given w) + \alpha_{0,j}^*(a_0 \given w)\right\} - 2  \right]  dP_0(w). \label{eqn::MSE1}
\end{align}
Observing that $\mathcal{F}_{\rm anti}$ is a convex cone, then for all $h \in \mathcal{F}_{\rm anti}^{1/\eta}$, the path $\{ t \mapsto (1-t) f_{n,0}^{(a_0)} + t h: t \in [0,1]\}$ lies entirely in $\mathcal{F}_{\rm anti}$ and passes through $f_{n,0}^{(a_0)}$ at $t=0$. Furthermore, since Condition \ref{cond::positivity2} implies the derivative of the risk for every element in the path is almost surely bounded, from the mean value theorem for differentiation and dominated convergence theorem we obtain
\begin{align*}
    & \lim_{t \downarrow 0} \frac{ R^{(a_0)} \big(f_{n,0}^{(a_0)}+t(h-f_{n,0}^{(a_0)})\big) - R^{(a_0)} \big(f_{n,0}^{(a_0)}\big)}{t} \\
    & = (-2) \sum_{j \in [J]} \int \ind(A=a_0) \big\{\big(h-f_{n,0}^{(a_0)}\big) \circ \pi_{n,j}(a_0 \given w)\big\} \big\{ f_{n,0}^{(a_0)} \circ \pi_{n,j}(a_0 \given w) - \pi_0^{-1}(a_0 \given w) \big\}  dP_0(a,w)\\
    & \geq 0.
\end{align*}
Thus, the minimizing solutions $\{\alpha_{0,j}^*: j \in [J]\}$ corresponding to the population risk satisfy, for all $h \in \mathcal{F}_{\rm anti}^{1/\eta}$, the inequality:
$$\sum_{j \in [J]} \int \left\{h \circ \pi_{n,j}(a_0 \given w) -  \alpha_{0,j}^*(a_0 \given w)\right\} \left\{\pi_0(a_0 \given w)\alpha_{0,j}^*(a_0 \given w) - 1 \right\} dP_0(w) \leq 0.$$
Taking $h$ such that $\alpha_{n,j}^*(a \given \cdot) : w \mapsto h \circ \pi_{n,j}(a \given w) $ in the above display, Equation \eqref{eqn::MSE1} implies that the excess risk satisfies the following lower bound:
\begin{align*}
    & \sum_{j \in [J]} \big\{ P_0\ell_{n,j}^{(a_0)}(f_n^{(a_0)}) - P_0\ell_{n,j}^{(a_0)}(f_{n,0}^{(a_0)}) \big\} \\
    & = \sum_{j \in [J]} \int \pi_0(a_0 \given w) \left[ \alpha_{n,j}^*(a_0 \given w) - \alpha_{0,j}^*(a_0 \given w) \right]^2   dP_0(w) \\
    & ~~ + \sum_{j \in [J]} \int \left[ \alpha_{n,j}^*(a_0 \given w) - \alpha_{0,j}^*(a_0 \given w) \right] \left[ 2\pi_0(a_0 \given w)\alpha_{0,j}^*(a_0 \given w)  - 2  \right]  dP_0(w)\\
    & \geq \sum_{j \in [J]} \int \pi_0(a_0 \given w) \left[ \alpha_{n,j}^*(a_0 \given w) - \alpha_{0,j}^*(a_0 \given w) \right]^2 dP_0(w) .
\end{align*}

Now, we obtain an upper bound for the excess risk. By Lemma \ref{lemma::nearERM}, the truncated empirical risk minimizer \(f_n^{(a_0)}\) is a near empirical risk minimizer satisfying
\[
\sum_{j \in [J]} \left\{ P_{n,j} \ell_{n,j}^{(a_0)}(f_n^{(a_0)}) - P_{n,j} \ell_{n,j}^{(a_0)}(f_{n,0}^{(a_0)}) \right\} \leq O_p(n^{-2/3}).
\]
Thus,
\begin{align*}
    & \sum_{j \in [J]} P_0\ell_{n,j}^{(a_0)}(f_n^{(a_0)}) -  \sum_{j \in [J]} P_0\ell_{n,j}^{(a_0)}(f_{n,0}^{(a_0)}) \\
    &=  \sum_{j \in [J]} \big\{P_0\ell_{n,j}^{(a_0)}(f_n^{(a_0)})  - P_{n,j}\ell_{n,j}^{(a_0)}(f_n^{(a_0)})\big\} + \sum_{j \in [J]}\big\{  P_n\ell_{n,j}^{(a_0)}(f_n^{(a_0)}) -  P_{n,j}\ell_{n,j}^{(a_0)}(f_{n,0}^{(a_0)})\big\}\\
    & \quad -  \sum_{j \in [J]}\big\{P_0\ell_{n,j}^{(a_0)}(f_{n,0}^{(a_0)})  - P_{n,j}\ell_{n,j}^{(a_0)}(f_{n,0}^{(a_0)})\big\}\\
    & \leq \sum_{j \in [J]} (P_0 - P_{n,j})\big\{\ell_{n,j}^{(a_0)}(f_n^{(a_0)})  - \ell_{n,j}^{(a_0)}(f_{n,0}^{(a_0)})\big\} + O_p(n^{-2/3}),
\end{align*}
where we used that $ \sum_{j \in [J]}\{  P_{n,j}\ell_{n,j}^{(a_0)}(f_n^{(a_0)}) -  P_{n,j}\ell_{n,j}^{(a_0)}(f_{n,0}^{(a_0)})\} \leq O_p(n^{-2/3})$.

Putting the lower and upper bounds together, we obtain the inequality:
\begin{equation}
\begin{aligned}
   & \sum_{j \in [J]} \int \pi_0(a_0 \mid w) \left[ \alpha_{n,j}^*(a_0 \mid w) - \alpha_{0,j}^*(a_0 \mid w) \right]^2 \, dP_0(w)  \\
   & \hspace{1cm} \leq  \sum_{j \in [J]} (P_0 - P_{n,j})\big\{\ell_{n,j}^{(a_0)}(f_n^{(a_0)}) - \ell_{n,j}^{(a_0)}(f_{n,0}^{(a_0)})\big\} + O_p(n^{-2/3}).
\end{aligned}
    \label{eqn::MSEBasic}
\end{equation}
Note, for each $j \in [J]$, that
$$\big\|\ell_{n,j}^{(a_0)}(f_n^{(a_0)}) - \ell_{n,j}^{(a_0)}(f_{n,0}^{(a_0)}) \big\| \lesssim \big\|\alpha_{n,j}^*(a_0 \mid \cdot) - \alpha_{0,j}(a_0 \mid \cdot)\big\|,$$
since $\alpha_{n,j}^*(a_0 \mid \cdot) $ and $\alpha_{0,j}(a_0 \mid \cdot)$ are uniformly bounded by $1/\eta$. Under \ref{cond::positivity2}, we have $\pi_0(a_0 \mid W) > \eta$ almost surely, thus,  
\begin{align*}
    \big\|\ell_{n,j}^{(a_0)}(f_n^{(a_0)}) - \ell_{n,j}^{(a_0)}(f_{n,0}^{(a_0)}) \big\|^2 &\lesssim \big\|\alpha_{0,j}^*(a_0 \mid w) - \alpha_{0,j}(a_0 \mid \cdot) \big\|^2  \\
& \lesssim \eta^{-1} \int \pi_0(a_0 \mid w) \big[ \alpha_{n,j}^*(a_0 \mid w) - \alpha_{0,j}^*(a_0 \mid w) \big]^2 \, dP_0(w) . 
\end{align*}
Hence, \eqref{eqn::MSEBasic} implies
\begin{equation}
    \delta_n^2  \leq \sum_{j \in [J]} (P_0 - P_{n,j})\big\{\ell_{n,j}^{(a_0)}(f_n^{(a_0)}) - \ell_{n,j}^{(a_0)}(f_{n,0}^{(a_0)})\big\} + Z_n \leq \sum_{j \in [J]} \phi_{n,j}^{(a_0)}(\delta_n) + Z_n,
    \label{eqn::MSEBasic2}
\end{equation}
where $Z_n$ is some random variable satisfying $Z_n = O_p(n^{-2/3})$.  Here, we define
\begin{align*}
    \delta_n^2 & \coloneq  \sum_{j \in [J]}\int \pi_0(a_0 \given w) \left[ \alpha_{n,j}^*(a_0 \given w) - \alpha_{0,j}^*(a_0 \given w) \right]^2 dP_0(w); \\
    \phi_{n,j}^{(a_0)}(\delta) & \coloneq \sup_{f \in \mathcal{F}_{n,j}^{(a_0)}: \|f\| \leq C \delta } (P_0 - P_{n,j})f; \\
    \mathcal{F}_{n,j}^{(a_0)} & \coloneq \big\{\ell_{n,j}^{(a_0)}(\theta_1)  - \ell_{n,j}^{(a_0)}(\theta_2) : \theta_1, \theta_2 \in \mathcal{F}_{\rm anti}^{1/\eta}\big\};
\end{align*}and $C > 0$ is a sufficiently large constant depending only on $J$, $M$, and $\eta$.

Almost surely, we know that the function class \(\mathcal{F}_{n,j}^{(a_0)}\) is uniformly bounded and is a Lipschitz transformation of \(\mathcal{F}_{n, {\rm anti}}^{(j,a_0),1/\eta} \times \mathcal{F}_{n, {\rm anti}}^{(j,a_0),1/\eta}\), the product of function classes defined by truncating \(\mathcal{F}_{n, {\rm anti}}^{(j,a_0)}\) with \(1/\eta\). By Theorem 2.10.20 of \cite{van1996weak} and arguing as in the proof of Theorem \ref{theorem::CAL}, we have
\begin{align*}
    \mathcal{J}\big(\delta, \mathcal{F}_{n,j}^{(a_0)} \big) \lesssim \mathcal{J}\big(\delta, \mathcal{F}_{n, {\rm anti}}^{(j,a_0),1/\eta}\big) \lesssim \mathcal{J}\big(\delta, \mathcal{F}_{n, {\rm anti}}^{(j,a_0)}\big) = \mathcal{J}\big(\delta, \mathcal{F}_{\rm anti}\big) \lesssim \sqrt{\delta}.
\end{align*}
Applying Theorem 2.1 in \cite{van2011local} conditional on \(\mathcal{E}^{(j)}\), we then have for each \(j \in [J]\),
\begin{align*}
    E_0^n \big[\phi_{n,j}^{(a_0)}(\delta) \given \mathcal{E}^{(j)} \big] \lesssim n^{-1/2} \mathcal{J}\big(\delta, \mathcal{F}_{n,j}^{(a_0)}\big) \Big(1 + \frac{\mathcal{J}(\delta,\mathcal{F}_{n,j}^{(a_0})}{\sqrt{n} \delta^2}\Big).
\end{align*}
Arguing similarly as in the proof of Theorem \ref{theorem::CAL}, since the right-hand side above is upper-bounded by a deterministic function of \(\mathcal{J}(\delta, \mathcal{F}_{\rm anti})\), we can show that the summation term satisfies
\begin{align*}
    \sum_{j=1}^J E_0^n[\phi_{n,j}^{(a_0)}(\delta)] \lesssim n^{-1/2} \mathcal{J}\big(\delta, \mathcal{F}_{\rm anti}\big) \Big(1 + \frac{\mathcal{J}(\delta, \mathcal{F}_{\rm anti})}{\sqrt{n} \delta^2}\Big).
\end{align*}

Proceeding exactly as in the proof of Theorem \ref{theorem::CAL}, we aim to show that \(\lim_{n\to\infty} P(n^{1/3} \delta_n > 2^S) \to 0\) as \(S \to \infty\). Let \(M_{\varepsilon} < \infty\) be such that \(|Z_n| \leq M_{\varepsilon} n^{-2/3}\) with probability at least \(1 - \varepsilon\). Denote the event \(A_s \coloneq \left\{ r_n^{-1} \delta_n \in (2^s, 2^{s+1}]\right\}\). By a standard peeling argument and Markov's inequality, we have for any rate \(\{r_n\}_{n \ge 0}\),
\begin{align*}
    \P\big(r_n^{-1} \delta_n > 2^S\big) &= \P\big(r_n^{-1} \delta_n > 2^S, |Z_n| \leq M_{\varepsilon} n^{-2/3}\big) + \varepsilon \\
    &= \sum_{s=S}^\infty \P\big(2^{s+1} \geq r_n^{-1} \delta_n > 2^s, |Z_n| \leq M_{\varepsilon} n^{-2/3}\big) + \varepsilon \\
    &= \sum_{s=S}^\infty \P\Big(A_s, \delta_n^2 \leq \frac{1}{J}\sum_{j \in [J]} \phi_{n,j}^{(a_0)}(\delta_n) + Z_n, |Z_n| \leq M_{\varepsilon} n^{-2/3}\Big) + \varepsilon \\
    &\leq \sum_{s=S}^\infty \P\Big(r_n^2 2^{2s} < \delta_n^2 \leq \frac{1}{J}\sum_{j \in [J]} \phi_{n,j}^{(a_0)}(r_n 2^{(s+1)}) + M_{\varepsilon} n^{-2/3}\Big) + \varepsilon \\
    &\leq \sum_{s=S}^\infty \P\Big(r_n^2 2^{2s} < \frac{1}{J}\sum_{j \in [J]} \phi_{n,j}^{(a_0)}(r_n 2^{(s+1)}) + M_{\varepsilon} n^{-2/3}\Big) + \varepsilon \\
    &\leq \sum_{s=S}^\infty \frac{J^{-1} \sum_{j \in [J]} \E\big[\phi_{n,j}^{(a_0)}(r_n 2^{(s+1)})\big] + M_{\varepsilon} n^{-2/3}}{r_n^2 2^{2s}} + \varepsilon \\
    &\leq \sum_{s=S}^\infty 2^{-2s} r_n^{-2} \left(n^{-1/2} \mathcal{J}\big(r_n 2^{(s+1)}, \mathcal{F}_{\rm anti}\big) \left(1 + \frac{\mathcal{J}(r_n 2^{(s+1)}, \mathcal{F}_{\rm anti})}{\sqrt{n} r_n^2 2^{2(s+1)}}\right) + M_{\varepsilon} n^{-2/3}\right) + \varepsilon \\
    &\lesssim \sum_{s=S}^\infty 2^{-2s} r_n^{-2} \left[n^{-1/2} \sqrt{r_n 2^{(s+1)}} \left(1 + \frac{\sqrt{r_n 2^{(s+1)}}}{\sqrt{n} r_n^2 2^{2(s+1)}}\right) + M_{\varepsilon} n^{-2/3}\right] + \varepsilon \\
    &\lesssim \sum_{s=S}^\infty \left[2^{-2s/3} r_n^{-3/2} n^{-1/2} + 2^{-2s} r_n^{-3} n^{-1} + 2^{-2s} r_n^{-2} M_{\varepsilon} n^{-2/3}\right] + \varepsilon.
\end{align*}
Choosing \(r_n = n^{-1/3}\), the above implies
\begin{align*}
    \lim_{n\to\infty} \P\big(r_n^{-1} \delta_n > 2^S\big) \lesssim \sum_{s=S}^\infty \frac{1}{2^{s}} + \varepsilon \overset{S \to \infty}{\to} \varepsilon.
\end{align*}
Since \(\varepsilon > 0\) was arbitrary, we conclude that for all \(a_0 \in \mathcal{A}\), \(\delta_n = O_p(n^{-1/3})\) and therefore
\begin{align*}
    \delta_n^2 = \sum_{j \in [J]} \int \pi_0(a_0 \given w) \big[ \alpha_{n,j}^*(a_0 \given w) - \alpha_{0,j}^*(a_0 \given w) \big]^2 dP_0(w) = O_p(n^{-2/3}).
\end{align*}
Summing the left-hand side over \(a_0 \in \mathcal{A}\), we obtain
\[
\sum_{j \in [J]} \int \big[ \alpha_{n,j}^*(a \given w) - \alpha_{0,j}^*(a \given w) \big]^2 dP_0(a,w) = O_p(n^{-2/3}),
\]
as desired. The second inequality follows from the above display after applying Minkowski's inequality and using that \(\{\alpha_{0,j}^*: j \in [J]\}\) minimizes the fold-averaged mean square error distance from \(\alpha_0\) over all monotone transformations of \(\{\pi_{n,j}(a_0 \given \cdot): j \in [J], a_0 \in \mathcal{A}\}\).

\section{Proof of Theorem \ref{theorem::AIPW}}

Denote the pseudo-outcome function:
\begin{equation}
   \chi_{0}: o = (w,a,y) \mapsto \mu_{0}(1, w) - \mu_{0}(0,w) + \Big\{  \frac{\ind(a=1)}{\pi_0(1 \given w)}  - \frac{\ind(a=0)}{\pi_0(0 \given w)} \Big\}\left\{y - \mu_0(a, w) \right\}.
\end{equation}
Note $\psi_0 = E_0[\chi_{0}(O)]$. For $j \in [J]$, let $o \mapsto \chi_{n,j}^*(o)$ denote the fold-specific estimator of the pseudo-outcome $\chi_0$ corresponding to the nuisance estimators $\alpha_{n,j}^*$ and $\mu_{n,j}$ of $\pi_0^{-1}$ and $\mu_0$.  We note that $\psi_n^* = \frac{1}{n}\sum_{i=1}^n \chi_{n,j(i)}^*(O_i)$.

Adding and subtracting, we have the von Mises bias expansion:
   \begin{align*}
       \psi_n^* - \psi_0 &= \frac{1}{J}\sum_{j=1}^J P_{n,j} \chi_{n,j}^*- P_0 \chi_0\\
       &=  (P_{n} - P_0)\chi_0  +  \frac{1}{J}\sum_{j=1}^J (P_{n,j} - P_0)(\chi_{n,j}^*- \chi_0) +  \frac{1}{J}\sum_{j=1}^J P_0(\chi_{n,j}^*- \chi_0).
   \end{align*}  
   To establish the desired result, it suffices to show that $(P_{n,j} - P_0)(\chi_{n,j}^*- \chi_0) = o_p(n^{-1/2})$ and $P_0(\chi_{n,j}^*- \chi_0) = o_p(n^{-1/2})$ for each $j \in [J]$. In the following, we implicitly condition on $\{O_1, \dots, O_n\}$ in our expectations. To show the latter bound, note by applying the law of total expectation and Hölder's inequality we get for each $j \in [J]$,
   \begin{align*}
       & P_0(\chi_{n,j}^*- \chi_0) \\
       &= \sum_{a_0 \in \{0,1\}} (-1)^{1+a_0} E_0\big[\mu_{n,j}(a_0,W) - \mu_0(a_0,W) + \ind(A=a_0)\alpha_{n,j}^*(a_0 \given W)\{Y - \mu_{n,j}(a_0,W)\} \\
       & ~~  - \ind(A=a_0)\pi_0^{-1}(a_0 \given W)\{Y - \mu_{0}(a_0,W)\} \big]\\
       &= \sum_{a_0 \in \{0,1\}} (-1)^{1+a_0} E_0\big[\mu_{n,j}(a_0,W) - \mu_0(a_0,W) + \ind(A=a_0)\alpha_{n,j}^*(a_0 \given W)\{\mu_0(a_0,W) - \mu_{n,j}(a_0,W)\} \big]\\
        &= \sum_{a_0 \in \{0,1\}} (-1)^{1+a_0} E_0\big[ \{\ind(A=a_0)\alpha_{n,j}^*(a_0 \given W) - 1 \}\{\mu_0(a_0,W) - \mu_{n,j}(a_0,W)\} \big]\\
         &= \sum_{a_0 \in \{0,1\}} (-1)^{1+a_0} E_0\big[ \ind(A=a_0) \big\{\alpha_{n,j}^*(a_0 \given W) - \pi_0^{-1}(a_0 \given W) \big\}\{\mu_0(a_0,W) - \mu_{n,j}(a_0,W)\} \big]\\
         & \leq \sum_{a_0 \in \{0,1\}} E_0\big[ \pi_0(a_0 \given W) \left\{\alpha_{n,j}^*(a_0 \given W) - \pi_0^{-1}(a_0 \given W) \right\}^2 \big] \lVert \mu_0(a_0, \cdot) - \mu_{n,j}(a_0, \cdot)\rVert \\
        & \lesssim \max_{a_0 \in \{0,1\}} \big\lVert \alpha_{n,j}^*(a \given \cdot) - \pi_0^{-1}(a \given \cdot) \big\rVert_{P_{0,a_0}} \lVert \mu_0(a_0, \cdot) - \mu_{n,j}(a_0, \cdot)\rVert;
   \end{align*} 
By Theorem \ref{theorem::MSE}, we have
\[
\sum_{j \in [J]}\|\alpha_{n,j}^*(a \given \cdot) - \pi_0^{-1}(a \given \cdot) \|_{P_{0}} \lesssim o_p(n^{-1/3}) + \|f_{n,0}^{(a)} \circ \pi_{n,j}(a \given \cdot) - \pi_0^{-1}(a \given \cdot)\|_{P_{0}},
\]
where $f_{n,0}^{(a)}$ is the optimal nondecreasing calibration for the cross-fitted propensity score estimator. Using symmetry and \ref{cond::AIPW2} and \ref{cond::AIPWDR}, we conclude that $ P_0(\chi_{n,j}^* - \chi_0) = o_p(n^{-1/2})$, as desired.

Next, we will show that $E_0^n \left[\lvert (P_{n,j} - P_0)(\chi_{n,j}^* - \chi_0)\rvert\right] = o(n^{-1/2})$, which implies, by Markov's inequality, that $(P_{n,j} - P_0)(\chi_{n,j}^* - \chi_0) = o_p(n^{-1/2})$. To this end, we denote the function class:
\begin{align*}
    & \Big\{ o \mapsto \mu_{n,j}(1, w) - \mu_{n,j}(0,w) + \big\{1(a_0 = 1)\theta \circ \pi_{n,j}(1 \given w) - 1(a_0 = 0)\theta \circ \pi_{n,j}(0 \given w)\big\}\{y - \mu_n(a_0, w)\} \\
    & \quad - \mu_{0}(1, w) + \mu_{0}(0,w) - \big\{1(a_0 = 1){\pi_0^{-1}(1 \given w)} - 1(a_0 = 0){\pi_0^{-1}(0 \given w)} \big\}\{y - \mu_0(a_0, w)\}: \theta \in \mathcal{F}_{\rm anti}^{1/\eta}\Big\}
\end{align*}
by $\mathcal{G}_{n,j}$. Of note, $\mathcal{G}_{n,j}$ is deterministic conditional on the training set $\mathcal{E}^{(j)}$. With probability tending to one, we have by definition and \ref{cond::positivity}, $\alpha_{n,j}^*(a_0 \given \cdot) = f_n^{(a_0)} \circ \pi_{n,j}(a_0 \given \cdot)$ where $f_n^{(a_0)} \in \mathcal{F}_{\rm anti}^{1/\eta}$. Hence, $\chi_{n,j}^* - \chi_0 \in \mathcal{G}_{n,j}$ with probability tending to one. Without loss of generality, we assume the event $\chi_{n,j}^* - \chi_0 \in \mathcal{G}_{n,j}$ occurs. Since by Conditions \ref{cond::boundedAIPW} and \ref{cond::AIPW2}, we know that $\mu_0(a, w)$ is bounded with probability tending to one for $a \in \{0,1\}$; combined with the analysis we did in the proof of Theorem \ref{theorem::CAL}, we know that the function class $\mathcal{G}_{n,j}$ is uniformly bounded with probability tending to one, and therefore a Lipschitz transformation of $\mathcal{F}_{n, \rm{anti}}^{(j,a_0),1/\eta} \times \mathcal{F}_{n, \rm{anti}}^{(j,a_0),1/\eta}$.

An argument identical to the proof of Theorem \ref{theorem::MSE} shows that, under \ref{cond::bound}-\ref{cond::positivity2}, that $\mathcal{J}(\delta, \mathcal{G}_{n,j}) \leq K\sqrt{\delta}$ for some fixed constant $K >0$. Let $\delta_n := \|\chi_{n,j}^*- \chi_0\|_{P_0}$; observe that by Minkowski's inequality and Hölder's inequality we have
\begin{align*}
    \delta_n & \leq \max_{a_0 \in \{0,1\}}\lVert\mu_{n,j}(a_0, \cdot) - \mu_0(a_0, \cdot)\rVert_{P_0} + \max_{a_0 \in \{0,1\}}\lVert\alpha_{n,j}^*(a_0 \given \cdot) - \pi_0^{-1}(a_0 \given \cdot)\rVert_{P_0} C_Y \\
    & ~~+ \max_{a_0 \in \{0,1\}}\big\lVert\alpha_{n,j}^*(a_0 \given \cdot) \mu_{n,j}(a_0, \cdot) - \pi_0^{-1}(a_0 \given \cdot) \mu_{n,j}(a_0, \cdot) + \pi_0^{-1}(a_0 \given \cdot) \mu_{n,j}(a_0, \cdot) -  \pi_0^{-1}(a_0 \given \cdot) \mu_0(a_0, \cdot)\big\rVert_{P_0} \\
    & \lesssim \max_{a_0 \in \{0,1\}}\|\alpha_{n,j}^*(a_0 \given \cdot) - \pi_0^{-1}(a_0 \given \cdot)\|_{P_0} + \max_{a_0 \in \{0,1\}}\|\mu_{n,j}(a_0 \given \cdot) - \mu_0(a_0 \given \cdot)\|_{P_0} = o_p(1),
\end{align*}
where we used Conditions \ref{cond::positivity}, \ref{cond::boundedAIPW}-\ref{cond::AIPW} and Theorem \ref{theorem::MSE}. Hence, we can find a deterministic sequence $\varepsilon_n \downarrow 0$ such that $\delta_n = o_p(\varepsilon_n)$. Without loss of generality, we can assume that $\delta_n \leq C \varepsilon_n$ for some fixed $C > 0$, since this occurs with probability tending to one. Working on this event, we have
\begin{align*}
    E_0^n \big[\big\lvert(P_{n,j} - P_0)(\chi_{n,j}^*- \chi_0)\big\rvert\big] \lesssim    E_0^n \Big[ \sup_{g \in \mathcal{G}_{n,j}: \lVert g \rVert \leq C \varepsilon_n} \big\lvert (P_{n,j} - P_0)g \big\rvert \Big].
\end{align*}
Noting $\mathcal{G}_{n,j}$ is uniformly bounded and that $\mathcal{J}(\varepsilon_n, \mathcal{G}_{n,j}) \rightarrow 0$ since $\mathcal{J}(1, \mathcal{G}_{n,j}) < \infty$, we can apply Theorem 2.1 of \cite{van2011local} conditional on $\mathcal{E}^{(j)}$ to conclude that:
\begin{align*}
     E_0^n \big[ \abs{(P_{n,j} - P_0)(\chi_{n,j}^*- \chi_0)} \given \mathcal{E}^{(j)} \big] \lesssim n^{-1/2}   \mathcal{J}(\varepsilon_n, \mathcal{G}_{n,j}) \Big(1 + \frac{\mathcal{J}(\varepsilon_n, \mathcal{G}_{n,j})}{\sqrt{n} \delta^2} \Big) = o(n^{-1/2}),
\end{align*}
where the right-hand side is nonrandom. Taking the expectation of both sides, we find 
\begin{align*}
    E_0^n \big[\big\lvert (P_{n,j} - P_0)(\chi_{n,j}^*- \chi_0) \big\rvert\big] = o(n^{-1/2})
\end{align*}
as desired. As the above bounds hold for all $j \in [J]$, we conclude $  \psi_n^* - \psi_0 = (P_n - P_0) \chi_0 + o_p(n^{-1/2})$. Noting that $\chi_0 - P_0 \chi_0$ is the efficient influence function of the ATE for a nonparametric model, we also conclude that $\psi_n^*$ is regular and efficient.

\end{document}